\documentclass{article} % For LaTeX2e
\usepackage{iclr2027_conference,times}
\usepackage{hyperref}
\usepackage{url}
\usepackage{amsmath,amssymb,amsthm,bm}   % mathematics
\usepackage{graphicx}                    % \includegraphics
\usepackage{booktabs}                    % \toprule, \midrule in the tables
\usepackage{cancel}                      % \cancel, in the derivations
\usepackage{microtype}                   % character protrusion and font expansion.
\usepackage{flafter}                     % a float never appears before its place in
\usepackage{algorithm}
\usepackage{algorithmic}
\usepackage{tikz}\usetikzlibrary{calc,patterns,decorations.pathreplacing,arrows.meta}
\usepackage[capitalise]{cleveref}        % must be loaded after hyperref
\graphicspath{{./}{figures/}}

\makeatletter
\ifdefined\cref\else
  \def\shim@sec{sec}\def\shim@app{app}\def\shim@tab{tab}\def\shim@fig{fig}
  \def\shim@eq{eq}\def\shim@lem{lem}\def\shim@prop{prop}
  \def\shim@one#1:#2\@nil{\def\shim@pfx{#1}%
    \ifx\shim@pfx\shim@sec Section~\ref{#1:#2}\else
    \ifx\shim@pfx\shim@app Appendix~\ref{#1:#2}\else
    \ifx\shim@pfx\shim@tab Table~\ref{#1:#2}\else
    \ifx\shim@pfx\shim@fig Figure~\ref{#1:#2}\else
    \ifx\shim@pfx\shim@eq Eq.~(\ref{#1:#2})\else
    \ifx\shim@pfx\shim@lem Lemma~\ref{#1:#2}\else
    \ifx\shim@pfx\shim@prop Proposition~\ref{#1:#2}\else
    \ref{#1:#2}\fi\fi\fi\fi\fi\fi\fi}
  \newcommand\shimrefs[1]{\def\shim@sep{}%
    \@for\@shimi:=#1\do{\shim@sep\def\shim@sep{, }\expandafter\shim@one\@shimi\@nil}}
  \DeclareRobustCommand{\cref}[1]{\shimrefs{#1}}
  \DeclareRobustCommand{\Cref}[1]{\shimrefs{#1}}
  \providecommand{\crefname}[3]{}
\fi
\makeatother

\newcommand{\R}{\mathbb{R}}
\newcommand{\T}{\mathbb{T}}
\newcommand{\Z}{\mathbb{Z}}

\newcommand{\x}{\bm{x}}

\newcommand{\G}{\bm{\Gamma}}

\DeclareMathOperator{\wrap}{wrap}

\newtheorem{proposition}{Proposition}
\newtheorem{lemma}{Lemma}
\newcommand{\ESS}{\mathrm{ESS}}
\newcommand{\KL}{\mathrm{KL}}
\newcommand{\Var}{\mathrm{Var}}

\newcommand{\hmu}{\hat\mu}
\newcommand{\hnu}{\hat\nu}
\newcommand{\hrho}{\hat\rho}
\newcommand{\M}{\bm{A}}

\DeclareMathOperator{\Tr}{Tr}

\newcommand{\Lfine}{4}                  % fine linear extent
\newcommand{\nplq}{1536}                % 6V oriented plaquettes
\newcommand{\dimBtwo}{765}              % 3V-3, free continuous dof
\newcommand{\suppiso}{13}               % median solve support, single-level 4^4
\newcommand{\suppisomax}{47}            % max solve support, single-level 4^4
\newcommand{\kappaSingleTwo}{45}        % condition number, single-level 2^4
\newcommand{\kappaSingleFour}{391}      % condition number, single-level 4^4
\newcommand{\Neval}{500{,}000}

\newcommand{\refPlqHalf}{0.245790 \pm 0.000043}   \newcommand{\refMonoHalf}{352.969 \pm 0.043}
\newcommand{\refPlqOne}{0.64944 \pm 0.00048}     \newcommand{\refMonoOne}{48.5 \pm 0.4}
\newcommand{\refPlqOneHalf}{0.813425 \pm 0.000044} \newcommand{\refMonoOneHalf}{0.411 \pm 0.004}
\newcommand{\refPlqTwo}{0.865396 \pm 0.000020}    \newcommand{\refMonoTwo}{0.0094 \pm 0.0007}
   
\newcommand{\fluxCost}{2\pi^2\beta}      % action cost of a unit flux through a 2-torus

\newcommand{\ZrecHalf}{1.03}
  
\newcommand{\monoHalf}{352.990 \pm 0.034}  
\newcommand{\KLone}{10.0}      \newcommand{\topOne}{0.043}
\newcommand{\ZrecOne}{1.3\times10^{-3}}
     
\newcommand{\KLsingle}{2.05}  \newcommand{\ESSsingle}{1.52\%}  \newcommand{\paramsSS}{3.91}
\newcommand{\KLmatched}{1.00} \newcommand{\ESSmatched}{4.65\%} \newcommand{\paramsML}{4.06}
\newcommand{\ndet}{771}                       % determined plaquettes, 4^4
\newcommand{\KLsingleHalf}{47.5}
\newcommand{\covSingleHalf}{8\times10^{-7}}   % math mode; = exp(log Zhat - log Z)
\newcommand{\covSingleTwo}{0.97}

\newcommand{\monoNullHalf}{419.4 \pm 0.3}
\newcommand{\monoNullTwo}{149.0 \pm 0.3}
\newcommand{\monoSingleHalf}{419.3}           % S050s, raw, N = 500,000
\newcommand{\monoSingleTwo}{0.0023}           % S_single, raw, N = 500,000
\newcommand{\monoSuppTwo}{6.6\times10^{4}}    % math mode; null / model at beta = 2

\newcommand{\AbbLayers}{48}     % their fine-lattice coupling layers
      
\newcommand{\cosSdOne}{0.033}        % per-configuration sd of <cos>, reference
\newcommand{\monoSdOne}{28}          % per-configuration sd of the monopole number
\newcommand{\hbMerge}{36}            % sweeps until hot and cold chains agree
\newcommand{\mcColdPlq}{0.6488}  \newcommand{\mcColdMono}{47.6}   % Metropolis in the 765 free angles
\newcommand{\mcHotPlq}{0.6546}   \newcommand{\mcHotMono}{45.2}
\newcommand{\monoRandom}{476}        % layer fed uniformly random free angles, 4^4

\newcommand{\teC}{0.188}

\newcommand{\KLoneHatModel}{3.40}

\newcommand{\tauQrange}{0.69\text{--}1.03}

\newcommand{\kappaScanLList}{4,6,8,10}               % 4D extents with kappa/V quoted
\newcommand{\kappaOverVList}{1.5,\,1.6,\,1.7,\,1.8}         % kappa / V at those extents
\newcommand{\kappaAlpha}{1.05}                 % 4D single level, kappa ~ |F|^alpha, fit over L = 4..10
\newcommand{\kappaFourSixteen}{1.3\times10^{5}}      % extrapolated 16^4
\newcommand{\kappaFourThirtytwo}{2.4\times10^{6}}   % extrapolated 32^4
\newcommand{\solveRowsUoneThirtytwo}{3.1\times10^{6}}   % 3V+3 at 32^4
\newcommand{\solveRowsSUthreeThirtytwo}{2.5\times10^{7}} % 8(3V+3) at 32^4
\newcommand{\kappaRefFour}{7}                  % one 4D refinement level, canonical coordinates
\newcommand{\kappaRefTwo}{1}                   % one 2D refinement level
\title{ Multilevel Plaquette-Space Sampling for Lattice Gauge Theories with Local Constraint Solves}

\author{%
Ankur Singha$^{1,2}$, Jacob Finkenrath$^{3}$, Karl Jansen$^{4,5}$,
Vipul Arora$^{6}$, Shinichi Nakajima$^{1,2,7}$ \\[2pt]
\normalfont
$^{1}$Berlin Institute for the Foundations of Learning and Data (BIFOLD), Germany \\
$^{2}$Technische Universit\"{a}t Berlin, Germany \\
$^{3}$Bergische Universit\"{a}t Wuppertal, Germany \\
$^{4}$Deutsches Elektronen-Synchrotron (DESY), Platanenallee 6, 15738 Zeuthen, Germany \\
$^{5}$Computation-Based Science and Technology Research Center, The Cyprus Institute, Nicosia, Cyprus \\
$^{6}$Department of Electrical Engineering (ESAT), KU Leuven, Belgium \\
$^{7}$RIKEN Center for Advanced Intelligence Project (AIP), Japan \\[2pt]
Correspondence: \texttt{a.singha@tu-berlin.de}, \texttt{nakajima@tu-berlin.de}%
}

\iclrfinalcopy

\begin{document}
\maketitle
% arXiv version: \iclrfinalcopy prints the authors, but the style also sets the running head
% "Published as a conference paper at ICLR 2027" inside \maketitle. The paper is under review,
% not accepted, so blank the head (and its rule) here, after \maketitle has set it.
\lhead{}\renewcommand{\headrulewidth}{0pt}
\begin{abstract}
% - Gauge theory
% - ML good for challenges: topological freezing
% - current ML models can't scale up: main challenge for this paper
% - sampling in link space
% - sampling in plaquette space: we tackle the challenges: Bianchi constraints
% - we propose
% - we show results for U1 and SU2.

% Lattice gauge theory is a class of field theory that is of great interest in nuclear, condensed matter and particle physics. Conventional Monte Carlo methods for sampling gauge theory lattices struggle with issues such as topological freezing. Machine learning-based methods have shown the potential to address these problems, albeit for smaller and low-dimensional lattices. Scaling up these methods is of great interest. 
% While current approaches define lattices in the link space, we propose that representing them in plaquette space is a better way, as the action (and hence, the distribution) depends directly on the plaquette and gauge redundancy is largely free. Further, we propose a multi-level sampling method to effectively satisfy Bianchi constraints, which restrict the plaquette space.
% Experiments on $U(1)$ in two and four dimensions and for $SU(2)$ in two dimensions support that the proposed approach is more effective than sampling link space lattices and single level sampling.

% Lattice gauge theory is a class of field theory that is of great interest in nuclear, condensed matter, and particle physics. 

Lattice gauge theories are an important class of physical models with high-dimensional structured distributions that underpin first-principles calculations in particle, nuclear, and condensed-matter physics. Recent advances in generative modeling have opened new avenues for sampling Boltzmann distributions in lattice gauge theories, offering the potential to alleviate limitations of traditional Monte Carlo methods, including critical slowing down and topological freezing. However, generative samplers for lattice gauge theories are typically constructed in link space, where correlations become increasingly long-ranged toward weak coupling, posing a challenge for learning. Plaquettes provide a more natural representation, as the action is local in these variables. However, exact Bianchi constraints restrict them to a lower-dimensional manifold, complicating direct generative modeling. We introduce a multilevel normalizing-flow construction that samples directly in plaquette space while satisfying these constraints exactly. The key idea is a coarse-to-fine factorization that transforms a globally coupled constraint problem into a sequence of local
solves: at each refinement, every determined plaquette depends on at most four newly generated variables, while the size of the constraint problem remains independent of the lattice size. We validate the construction for $U(1)$ in two and four dimensions and for $SU(2)$ in two dimensions. Our  multilevel
Plaquette-Space Sampler (PSS) substantially outperforms link-space baselines, with the advantage increasing toward weak coupling, where the gauge coupling becomes small. This regime is particularly challenging for generative sampling and, in asymptotically free gauge theories, is relevant to continuum studies.

\end{abstract}

\section{Introduction}
\label{sec:intro}
Lattice gauge theory provides a first-principles framework for studying strongly coupled
quantum gauge theories, where perturbative methods are no longer applicable
\citep{PhysRevD.10.2445,Gattringer:2010zz}. Numerical lattice simulations have become an
essential tool for quantitative predictions in particle and nuclear physics, ranging from
hadron spectroscopy and matrix elements to finite-temperature phenomena and precision tests
of the Standard Model. These calculations reduce physical observables to expectation values
over a 
 % \textcolor{blue}{[extremely high-dimensional in ML means more than billions]}
 high-dimensional probability distribution, making efficient sampling of
gauge-field configurations a central computational problem.

In practice, these expectations are estimated predominantly with Markov chain Monte Carlo (MCMC)
methods. Their computational cost, however, grows rapidly as the lattice spacing is reduced
toward the continuum limit and correlations extend over increasingly many lattice sites
\citep{wolff1990critical,Schaefer2011}; generative models have been applied directly to this
critical slowing down \citep{Albergo2019,PhysRevLett.122.080602,Kanwar_2020}. The problem is particularly severe for topological
observables, whose autocorrelation times can grow dramatically and eventually lead to
effective freezing in a fixed topological sector
\citep{DelDebbio2004,Luscher2011}. Improving the scalability of lattice sampling is therefore
important not only computationally, but also for extending reliable first-principles
predictions into increasingly demanding physical regimes.

Generative models provide a promising alternative to conventional Markov-chain updates.
A model with tractable density $q$ can generate global, approximately independent proposals,
while the mismatch with the target distribution $p\propto e^{-S}$ can be corrected exactly
through importance reweighting or an independence Metropolis step
\citep{Albergo2019,nicoli_pre}. Recent progress has shown that generative samplers can be
particularly effective at mitigating topological freezing in low-dimensional lattice gauge
theories \citep{Kanwar_2020,Rothkopf2026apb,Singha:2026aac}.
 These results highlight their potential, but also
shift the central challenge toward \emph{scalability}: applying the same ideas to larger
lattices and higher-dimensional gauge theories requires generative models whose quality does
not deteriorate rapidly with lattice size, correlation length, or weak coupling.

Most existing generative approaches sample configurations in \emph{link space}, including
gauge-equivariant flows \citep{Kanwar_2020,Boyda2021,Abbott:2023thq}, continuous and stochastic flows
\citep{Bacchio:2022vje,Caselle:2022acb,bonanno2026scalable,gerdes2023learning}, and diffusion models
\citep{Wang:2023exq,Zhu:2025pmw}. Recent stochastic-flow approaches aim
to mitigate topological freezing in four-dimensional gauge theories
%,
%while remaining in link space, combining 
%learned 
by learning transformations with long non-equilibrium Monte Carlo trajectories, so that much of the computational cost still comes from repeated MCMC updates~\citep{bonanno2026scalable}.  
% In all these cases, the gauge action, however, depends on the links only through
% the plaquettes, the ordered products of links around elementary squares. Plaquettes are
% therefore a natural representation of the physical degrees of freedom: the action is local
% in them and gauge redundancy is removed \textcolor{red}{up to one global transformation}. In link space, by contrast, every plaquette is a
% derived quantity involving several links, while each link participates in several
% plaquettes. As the coupling weakens, the plaquettes become increasingly concentrated near
% the identity and correlations extend over longer distances, requiring increasingly precise
% coordination among the link variables. Consistent with this picture, link-space flows have
% been observed to degrade toward weak coupling \citep{Abbott2023}.
However, the Wilson gauge action, which provides the standard discretization of the gauge-field action, depends on the link variables only through
the \emph{plaquettes}---the ordered products of links around elementary squares.%
\footnote{More general improved gauge actions may additionally depend on larger Wilson loops such as rectangular loops, to reduce discretization errors while retaining locality.}
Plaquettes are
therefore a natural representation of the physical degrees of freedom: the action is local
in them and gauge redundancy is removed up to one global transformation. In link space, by contrast, every plaquette is a
derived quantity involving several links, while each link participates in several
plaquettes. As the coupling weakens, the plaquettes become increasingly concentrated near
the identity and correlations extend over longer distances, requiring increasingly precise
coordination among the link variables. Consistent with this picture, link-space flows have
been observed to degrade toward weak coupling \citep{Abbott2023}.

This motivates generating plaquettes directly.  Indeed, a recent study reported successful results using plaquette-space sampling for two-dimensional $U(1$) theory \citep{Singha:2026aac}.  %This 
The present
work aims to generalize this approach to higher-dimensional theories and larger gauge groups.
% The difficulty is that plaquettes are not independent: exact Bianchi identities restrict physically admissible plaquette fields to a lower-dimensional constraint manifold. In two dimensions this constraint is essentially global, whereas in higher dimensions it becomes local and extensive. In four dimensions, for example, asymptotically half of the plaquettes are determined by the others, and for non-Abelian groups the corresponding relations involve parallel transporters. A generic density over all plaquettes therefore generates configurations outside the physical manifold. Soft penalties leave residual constraint violations, while projecting samples onto the manifold generally destroys the tractable density required for exact reweighting.
The difficulty is that a general plaquette configuration does not necessarily correspond to a valid link configuration: a valid plaquette configuration must satisfy geometric consistency constraints, known as the \emph{Bianchi identities} \citep{Batrouni1982}, which restrict physically admissible plaquette fields to a lower-dimensional constraint manifold. While these constraints are simple in two dimensions, generating plaquette configurations that satisfy all Bianchi identities is non-trivial in higher-dimensions.
Furthermore, 
%in higher-dimensional theories, 
the size of the resulting linear system can grow rapidly with the lattice size, making it increasingly challenging to solve.
%for large lattices.

% In four dimensions the system has $3V+3$ rows per generator, $\solveRowsSUthreeThirtytwo$ for $SU(3)$ on a $32^4$ lattice, and its condition number grows linearly with the lattice volume (\cref{app:solve_size}).

We tackle this issue using multilevel generative sampling. At the coarsest level, a small number of plaquette variables, together with additional degrees of freedom called \emph{holonomies}, are generated such that a subset of the Bianchi identities is satisfied. At each subsequent finer level, each coarse plaquette is decomposed into finer plaquettes while preserving its original value as the value of the corresponding Wilson loop at the finer level. This Wilson loop encloses multiple fine plaquettes, and its value is given by the ordered product of the plaquette variables contained within it.
% Consequently, all Bianchi identities satisfied at coarser levels remain satisfied at the finest level.
Consequently, all Bianchi identities satisfied at coarser levels are automatically inherited by the finer levels, and hence remain satisfied at the finest level.
 This transforms the globally coupled constraint problem into a sequence of local constraint problems: each determined plaquette depends on at most four newly generated variables, 
and the size of each local problem is 
 independent of the lattice size. 
 % (condition number $\kappaRefFour$ at every extent, \cref{app:solve_size}).
 Furthermore, multilevel sampling naturally captures long-range correlations through the coarse-level samplers and has been shown to be advantageous 
\citep{Singha2025RiGCS,bauer2025super,singha2026scalable,Hasenfratz2026mae,Hasenfratz2026azk} for studying physical systems, particularly near criticality and towards the continuum limit---regimes of particular physical relevance.

We instantiate the method for $U(1)$ gauge theory in two and four dimensions and for $SU(2)$ in two dimensions.
Numerical experiments show that  
our multilevel Plaquette-Space Sampler (PSS) substantially outperforms both single-level and multilevel link-space baselines.
In two-dimensional compact $U(1)$, PSS scales more favorably with lattice size than the corresponding baselines,
while 
in four-dimensional $U(1)$ and two-dimensional $SU(2)$, it
becomes increasingly advantageous toward weak coupling. 

The main contributions of this paper are:
\begin{itemize}
    \item 
    We generalize plaquette-space sampling to four dimensions and to $SU(2)$ gauge theory.
    \item 
    We propose a multilevel sampler for gauge theories that uses simple local procedures to satisfy the Bianchi identities while better capturing long-range correlations.
\item 
    We empirically demonstrate the benefits of both plaquette-space sampling and the multilevel scheme in physically relevant regimes where long-range correlations are pronounced.
\end{itemize}
\section{Background}
\label{sec:background}

\subsection{Lattice gauge theory}
\label{sec:bg.lgt}

\paragraph{The lattice and its variables.}
We consider a periodic hypercubic lattice in $d$ dimensions with $V$ sites, each denoted by $x$, and unit
vectors $\hmu$, $\mu=0,\dots,d-1$. Gauge fields live on the lattice edges: the oriented link
from $x$ to $x+\hmu$ carries an element $U_\mu(x)$ of a compact group $G$, while the opposite
orientation, i.e., from $x+\hmu$ to $x$, is $U_\mu(x)^\dagger$ (\cref{fig:gauge}a). The choice of $G$
specifies the gauge theory. When G is $U(1)$, $U=e^{i\theta}$ with
$\theta\in(-\pi,\pi]$; and when G is $SU(2)$,
$U=a_0+i\,\bm a\cdot\bm\sigma$, with $(a_0,\bm a)$ a unit vector in $\R^4$ and $\bm\sigma$ are $2\times2$ Pauli matrices. A lattice configuration is thus
a collection of link variables $U_\mu(x)$, i.e.\ a point in $G^{dV}$.

% Each compact group is equipped with its normalised Haar measure
% $\mathrm dU$, invariant under left and right multiplication.

\paragraph{Plaquettes and distribution.}
The elementary gauge-covariant object is the plaquette, the ordered product of links around
a unit square in the $\mu\nu$ plane (\cref{fig:gauge}a),
\begin{equation}
  P_{\mu\nu}(x)=U_\mu(x)\,U_\nu(x+\hmu)\,U_\mu(x+\hnu)^\dagger\,U_\nu(x)^\dagger .
  \label{eq:plq}
\end{equation}
The Wilson action is a sum of single-plaquette terms, and the target distribution is
\begin{align}
%  \begin{gathered}
  S(U)&=- \textstyle \beta\sum_{x,\,\mu<\nu}\tfrac1{N_c}\operatorname{Re}\operatorname{tr}P_{\mu\nu}(x),
  %\\[2pt]
 & p(U)\,\mathrm dU&= \textstyle  \frac1Z\,e^{-S(U)}\prod_{x,\mu}\mathrm dU_\mu(x),
%  \end{gathered}
  \label{eq:action}
\end{align}
where $N_c$ is the dimension of the matrices, $\mathrm dU$ the normalised Haar measure of $G$,
and $Z$ the partition function.
%$p(U)$ is the distribution we want to sample. 
The coupling
$\beta$ is the only parameter of the model.

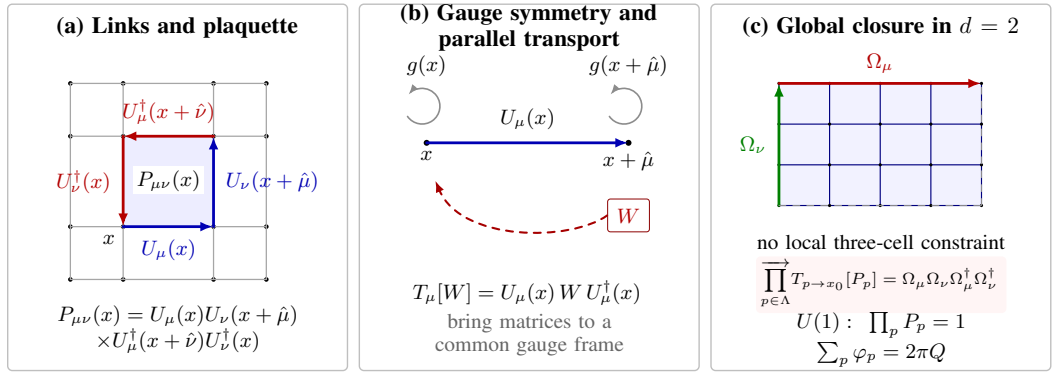
\begin{figure*}[tbh]
    \centering
    \resizebox{\textwidth}{!}{%
        \begin{tikzpicture}[
    >=Latex,
    line cap=round,
    line join=round,
    every node/.style={font=\small},
    panel/.style={draw=black!25,rounded corners=2pt,line width=0.6pt},
    gridline/.style={draw=black!35,line width=0.55pt},
    bluearrow/.style={
        draw=blue!70!black,
        line width=1.2pt,
        -{Latex[length=2.0mm]}
    },
    redarrow/.style={
        draw=red!70!black,
        line width=1.2pt,
        -{Latex[length=2.0mm]}
    },
    greenarrow/.style={
        draw=green!50!black,
        line width=1.2pt,
        -{Latex[length=2.0mm]}
    }
]

\def\PW{5.15}
\def\PH{5.55}
\def\GAP{0.18}

% ============================================================
% (a) Links and plaquette
% ============================================================
\begin{scope}[shift={(0,0)}]

\draw[panel] (0,0) rectangle (\PW,\PH);

\node[
    font=\bfseries,
    align=center,
    text width=4.7cm
]
at (2.575,5.20)
{(a) Links and plaquette};

\fill[blue!7]
    (1.75,2.18) rectangle (3.12,3.55);

\foreach \x in {0.95,1.75,3.12,3.92}{
    \foreach \y in {1.38,2.18,3.55,4.35}{
        \fill (\x,\y) circle (0.035);
    }
}

\foreach \y in {1.38,2.18,3.55,4.35}{
    \draw[gridline] (0.95,\y) -- (3.92,\y);
}

\foreach \x in {0.95,1.75,3.12,3.92}{
    \draw[gridline] (\x,1.38) -- (\x,4.35);
}

\draw[bluearrow]
    (1.75,2.18) -- (3.12,2.18);

\draw[bluearrow]
    (3.12,2.18) -- (3.12,3.55);

\draw[redarrow]
    (3.12,3.55) -- (1.75,3.55);

\draw[redarrow]
    (1.75,3.55) -- (1.75,2.18);

\node[below left,font=\footnotesize]
    at (1.75,2.18)
    {$x$};

\node[blue!70!black,below,font=\footnotesize]
    at (2.44,2.11)
    {$U_\mu(x)$};

\node[blue!70!black,right,font=\footnotesize]
    at (3.18,2.86)
    {$U_\nu(x+\hat\mu)$};

\node[red!70!black,above,font=\footnotesize]
    at (2.44,3.61)
    {$U_\mu^\dagger(x+\hat\nu)$};

\node[red!70!black,left,font=\footnotesize]
    at (1.69,2.86)
    {$U_\nu^\dagger(x)$};

\node[
    fill=white,
    inner sep=1.5pt
]
at (2.44,2.88)
{$P_{\mu\nu}(x)$};

\node[
    align=center,
    font=\footnotesize
]
at (2.58,0.63)
{
$P_{\mu\nu}(x)
 =U_\mu(x)U_\nu(x+\hat\mu)$\\
$\times U_\mu^\dagger(x+\hat\nu)U_\nu^\dagger(x)$
};

\end{scope}

% ============================================================
% (b) Gauge symmetry and parallel transport
% ============================================================
\begin{scope}[shift={({\PW+\GAP},0)}]

\draw[panel] (0,0) rectangle (\PW,\PH);

\node[
    font=\bfseries,
    align=center,
    text width=4.7cm
]
at (2.575,5.20)
{(b) Gauge symmetry and\\parallel transport};

\coordinate (X) at (1.02,3.45);
\coordinate (Y) at (4.08,3.45);

\fill (X) circle (0.045);
\fill (Y) circle (0.045);

\draw[bluearrow]
    (X) -- (Y);

\node[above,font=\footnotesize]
    at (2.55,3.53)
    {$U_\mu(x)$};

\node[below,font=\footnotesize]
    at (X)
    {$x$};

\node[below,font=\footnotesize]
    at (Y)
    {$x+\hat\mu$};

\node[font=\footnotesize]
    at (1.02,4.62)
    {$g(x)$};

\node[font=\footnotesize]
    at (4.08,4.62)
    {$g(x+\hat\mu)$};

\draw[
    black!45,
    line width=0.7pt,
    -{Latex[length=1.6mm]}
]
(0.73,3.85)
arc[start angle=215,end angle=500,radius=0.27];

\draw[
    black!45,
    line width=0.7pt,
    -{Latex[length=1.6mm]}
]
(3.79,3.85)
arc[start angle=215,end angle=500,radius=0.27];

\node[
    draw=red!55!black,
    rounded corners=1.2pt,
    inner sep=4pt,
    text=red!70!black,
    font=\footnotesize
]
(W) at (4.08,2.35)
{$W$};

\draw[
    red!65!black,
    dashed,
    line width=0.9pt,
    -{Latex[length=1.8mm]}
]
(W.west)
.. controls (3.2,1.95) and (1.55,1.95) ..
(1.15,2.88);

\node[
    align=center,
    font=\footnotesize
]
at (2.58,1.18)
{
$T_\mu[W]
 =U_\mu(x)\,W\,U_\mu^\dagger(x)$
};

\node[
    align=center,
    text width=4.35cm,
    font=\footnotesize,
    text=black!60
]
at (2.58,0.55)
{bring matrices to a common gauge frame};

\end{scope}

% ============================================================
% (c) Global closure in d=2
% ============================================================
\begin{scope}[shift={({2*(\PW+\GAP)},0)}]

\draw[panel] (0,0) rectangle (\PW,\PH);

\node[
    font=\bfseries,
    align=center,
    text width=4.7cm
]
at (2.575,5.20)
{(c) Global closure in $d=2$};

\fill[blue!5]
    (1.03,2.50) rectangle (4.10,4.35);

\foreach \i in {0,...,4}{
    \draw[
        blue!50!black,
        line width=0.5pt
    ]
    ({1.03+0.7675*\i},2.50)
    --
    ({1.03+0.7675*\i},4.35);
}

\foreach \j in {0,...,3}{
    \draw[
        blue!50!black,
        line width=0.5pt
    ]
    (1.03,{2.50+0.6167*\j})
    --
    (4.10,{2.50+0.6167*\j});
}

\foreach \i in {0,...,4}{
    \foreach \j in {0,...,3}{
        \fill
        ({1.03+0.7675*\i},{2.50+0.6167*\j})
        circle (0.024);
    }
}

\draw[redarrow]
    (1.03,4.35) -- (4.10,4.35);

\node[
    red!70!black,
    above,
    font=\footnotesize
]
at (2.56,4.39)
{$\Omega_\mu$};

\draw[greenarrow]
    (1.03,2.50) -- (1.03,4.35);

\node[
    green!50!black,
    left,
    font=\footnotesize
]
at (0.96,3.42)
{$\Omega_\nu$};

\draw[dashed,black!45]
    (4.10,2.50) -- (4.10,4.35);

\draw[dashed,black!45]
    (1.03,2.50) -- (4.10,2.50);

\node[
    align=center,
    text width=4.35cm,
    font=\footnotesize
]
at (2.58,1.95)
{no local three-cell constraint};

\node[
    fill=red!4,
    rounded corners=2pt,
    inner sep=2pt,
    align=center,
    font=\scriptsize
]
at (2.58,1.30)
{
$\displaystyle
\overrightarrow{\prod_{p\in\Lambda}}
T_{p\to x_0}[P_p]
=
\Omega_\mu\Omega_\nu
\Omega_\mu^\dagger\Omega_\nu^\dagger
$
};

\node[
    align=center,
    font=\footnotesize
]
at (2.58,0.47)
{
$U(1):\ \ \prod_p P_p=1$
\\[1pt]
$\sum_p\varphi_p=2\pi Q$
};

\end{scope}

\end{tikzpicture}
    }
    \caption{
    Links, gauge symmetry and the plaquette constraints.
    \textbf{(a)}~Oriented links and the plaquette $P_{\mu\nu}(x)$ they form.
    \textbf{(b)}~Independent gauge transformations at neighbouring sites, and the parallel
    transport of $W$ from $x+\hat\mu$ to the frame at $x$.
    \textbf{(c)}~In $d=2$ the only constraint is the global closure through the holonomies
    $\Omega_\mu$ and $\Omega_\nu$.
    }
    \label{fig:gauge_bianchi}\label{fig:gauge}
\end{figure*}

\paragraph{Gauge symmetry and gauge fixing.}
The link variables carry a local redundancy. A gauge transformation assigns an element
$g(x)\in G$ to every site and acts as
\begin{equation}
    U_\mu(x)\mapsto g(x)\,U_\mu(x)\,g(x+\hmu)^\dagger .
    \label{eq:gauge}
\end{equation}
A plaquette is a closed loop based at $x$, so along it the factors $g$ cancel in neighbouring
pairs and only the endpoint survives: $P_{\mu\nu}(x)\mapsto g(x)\,P_{\mu\nu}(x)\,g(x)^\dagger$.
Its trace is therefore unchanged and so is the action \eqref{eq:action}, while the Haar measure
is invariant under multiplication by a fixed group element.  Consequently, the whole distribution
\eqref{eq:action} is gauge invariant. The redundancy can be removed by fixing the $V-1$ links of
a maximal spanning tree to the identity, leaving one global transformation at the root
(\cref{app:dof.tree}, \cref{fig:dof}).
Gauge-invariant observables can be computed directly from the gauge-fixed samples.  For gauge-dependent observables, the gauge degrees of freedom can be sampled independently from the uniform distribution.

Gauge covariance also dictates how objects attached to different sites may be combined. A
matrix $W$ attached to the site $x+\hmu$ carries the gauge freedom of that site, transforming as
$W\mapsto g(x+\hmu)\,W\,g(x+\hmu)^\dagger$, so it cannot simply be multiplied with a matrix
attached to $x$. Conjugating it with the link that joins the two sites,
%\begin{equation}
$    T_\mu[W]=U_\mu(x)\,W\,U_\mu(x)^\dagger $,
%    \label{eq:transport}
%\end{equation}
gives an object that transforms at $x$ instead, because the factors $g(x+\hmu)$ carried by the
link \eqref{eq:gauge} cancel against those of $W$. This is called \emph{parallel transport}
(\cref{fig:gauge}b): the link is what carries an object to its neighbouring site. For $U(1)$ it
has no effect, 
beucase the group is Abelian; for $SU(2)$ it is essential for making products of plaquettes based at different sites well defined, 
and thus for
defining the constraints on plaquette space via Bianchi identity,
as explained below
(See also \cref{app:parallel_transport} for more details).

\paragraph{The Bianchi identity.}
%Plaquettes are gauge invariant, but they are not independent. 
Any configuration of link
variables produces plaquettes that satisfy a fixed set of relations exactly and automatically,
 because each link occurs in them once with each orientation. 
Accordingly, a sampler 
%that proposes links therefore never has to consider these relations, whereas a sampler 
that generates
plaquettes directly must satisfy all such constraints,
% ; that requirement is what the present work is
% built around.
% The relations all follow from one statement. Take any closed surface $\Sigma$ assembled from
% plaquettes, transport them to a common base point $x_0$ with \eqref{eq:transport} and multiply
% them in surface order; then the links cancel in pairs and leave
% \begin{equation}
%     \overrightarrow{\prod_{p\in\Sigma}}
%     T_{p\to x_0}\!\left[P_p^{\epsilon_{\Sigma p}}\right]
%     =\mathbf 1 ,
%     \label{eq:bianchi_surface}
% \end{equation}
% where $\epsilon_{\Sigma p}=\pm1$ orients each face consistently with $\Sigma$. This is the
% lattice Bianchi identity, the discrete counterpart of $DF=0$. The transports and the ordering
% matter only in the non-Abelian case, and a derivation is given in \cref{app:bianchi}.
%All constraints are given 
collectively expressed by the
lattice Bianchi identities.
For any contractible closed surface $\Sigma$,
the Bianchi identities imply
\begin{align}
       \textstyle \overrightarrow{\prod}_{p\in\Sigma}
  T_{p\to x_0}\!\left[P_p^{\epsilon_{\Sigma p}}\right]
    =\mathbf 1 ,
    \label{eq:bianchi_surface}
\end{align}
% assembled from plaquettes.
where the arrow indicates that the plaquettes are multiplied in the prescribed surface order, $T_{p\to x_0}[\cdot]$ transports each plaquette to a common base point $x_0$, and $\epsilon_{\Sigma p}=\pm1$ specifies its orientation relative to  $\Sigma$.
%The transports and the ordering matter only in the non-Abelian case, and a derivation is given in 
% Which surfaces exist depends on the dimension. On a periodic lattice, each coordinate plane
% carries a closed but 
For any closed non-contractible surface $\Lambda$ arising from the toroidal topology induced by periodic boundary conditions, the Bianchi identities instead imply
\begin{equation}
\textstyle 
    \overrightarrow{\prod}_{p\in\Lambda}
    T_{p\to x_0}[P_p]=\Omega_\mu\,\Omega_\nu\,\Omega_\mu^\dagger\,\Omega_\nu^\dagger,
    \label{eq:global_closure_2d}
\end{equation}
where $\Omega_\mu$ and $\Omega_\nu$ are the holonomies wrapping the $\mu$- and $\nu$-directions, which span $\Lambda$, respectively
(\cref{fig:gauge_bianchi}c).
% holonomies for the two directions that 
% Transporting around it does not return to the same frame, and \eqref{eq:bianchi_surface} instead
% leaves the commutator of the two holonomies $\Omega_\mu,\Omega_\nu$ that wrap the plane
% (\cref{fig:gauge_bianchi}c),
See \cref{app:bianchi} for derivation,
and \cref{app:bianchi.cases}  for the explicit forms of the Bianchi identities for $U(1)$ and $SU(2)$ theories.

\subsection{Related Work: Sampling with generative models}
\label{sec:bg.gen}

Expectation values under \eqref{eq:action} are conventionally estimated using MCMC methods. As correlations grow, however, local updates become increasingly correlated
and the cost of obtaining effectively independent configurations rises
\citep{DUANE1987216,wolff1990critical,Schaefer2011}. Generative models provide a different
route: a model with tractable density $q$ can propose independent configurations, while
deviations from the target distribution are corrected statistically~\citep{Albergo2019,nicoli_prl}.

% For samples drawn from $q$, self-normalised importance sampling uses weights
% $w=e^{-S}/q$ and is consistent whenever $q$ covers the support of $p$
% \citep{nicoli_pre,nicoli_prl,OwenMCbook}. The same model can instead be used as an
% independence Metropolis proposal with acceptance probability $\min(1,w'/w)$, yielding an
% exact Markov chain \citep{Albergo2019}. In either case, the generative model need not reproduce
% the target exactly; it provides a proposal whose residual error is corrected by the sampling
% procedure.

% We train the models by minimising the reverse Kullback--Leibler divergence
% \citep{10.1214/aoms/1177729694}. This objective requires only samples from $q$ and evaluations
% of the action, and therefore allows data-free training \citep{Vaitl_2022}. Its main
% limitation is its mode-seeking character: a model can assign little probability to important
% regions of the target and still obtain a favourable training objective
% \citep{Hackett:2021idh,Nicoli:2023qsl}. The problem becomes increasingly severe with system
% size. For local theories the divergence is extensive in the volume, and the number of
% importance samples required to reliably probe the target can grow as $e^{\KL}$
% \citep{chatterjee2018sample}.

\paragraph{Generative sampling for lattice field theory.}
Normalizing-flow samplers were first introduced for scalar lattice field theories
\citep{Albergo2019} and subsequently extended to gauge theories through exactly
gauge-equivariant coupling layers, first for $U(1)$ \citep{Kanwar_2020}, then for
$SU(N)$ \citep{Boyda2021}, and later to arbitrary space-time dimension
\citep{Abbott:2023thq}; see \citet{albergo2021introduction} and
\citet{Cranmer:2023xbe} for reviews. The framework has since been applied to fermionic
theories \citep{albergoflowbased,Albergo:2022qfi}, the Hubbard model
\citep{Schuh:2026dvp,Kreit:2026eng}, entanglement observables
\citep{Bulgarelli:2024yrz}, and correlated ensemble generation for QCD
\citep{Abbott:2024kfc,Abbott:2026ylv}. Related approaches include stochastic normalizing
flows based on Jarzynski's equality
\citep{PhysRevD.94.034503,Caselle:2022acb,Bulgarelli:2024brv}, learned continuous maps
\citep{Bacchio:2022vje,Gerdes:2024rjk}, autoregressive models  \citep{PhysRevLett.122.080602}, and diffusion
models for gauge theories
\citep{Wang:2023exq,Zhu:2025pmw,vega2025group,aarts2026generalizableequivariantdiffusionmodels,Kanwar:2025wuc}.
Very recently, \citet{Singha:2026aac} applied plaquette-space sampling to two dimensional $U(1)$ theory, and proposed a mixture-model approach to mitigate topological freezing.
All these previous works, except \citet{Singha:2026aac}, generate link variables directly.  Our work generalizes the plaquette-space sampling to higher dimensions and larger gauge groups.

% A central challenge across these approaches is scaling. Flow quality is known to deteriorate
% with volume \citep{Abbott:2022zsh}, and reverse-KL training can collapse onto only part of a
% multimodal target \citep{Hackett:2021idh,Nicoli:2023qsl}. For local theories, the
% extensivity of $\KL$, together with importance-sampling sample-complexity bounds
% \citep{chatterjee2018sample}, makes this deterioration particularly severe at large volume.
% Exact reweighting remains possible when the proposal has sufficient support
% \citep{nicoli_pre,nicoli_prl}, but an accurate and scalable proposal is still essential for
% practical efficiency. Our work addresses this problem at the level of representation and
% constraint structure rather than by introducing a new density model: the exact constraint
% layer of \cref{sec:m.build} can be combined with different tractable generative models, and
% we use spline couplings \citep{Durkan2019} as a convenient choice for circular variables.

\paragraph{Hierarchical and constrained generative models.}
Coarse-to-fine sampling has a long history, including multigrid Monte Carlo
\citep{Goodman:1986pv,Goodman:1989jw}, renormalization-group-guided updates
\citep{Schmidt:1983rng,FAAS1986571}, multiscale equilibration
\citep{Endres:2015yca}, and multilevel variance-reduction methods
\citep{PhysRevD.93.094507,Giles_2015}. The learned hierarchical constructions have directly connected generative modeling with the renormalization group
\citep{PhysRevLett.121.260601,Koch-Janusz2018,Hu_2022}. For discrete spin systems, hierarchical generative approaches include HAN \citep{BIALAS2022108502} and RiGCS \citep{Singha2025RiGCS}, the latter constructing configurations from coarse to fine by conditioning each refinement level on the preceding one. For continuous lattice field theories, closely related coarse-to-fine generative sampling has been developed through super-resolving normalizing flows \citep{bauer2025super} and multilevel conditional-flow constructions \citep{singha2026scalable}. In lattice gauge theory, \citet{Abbott2023} proposed a multiscale normalizing-flow construction that progressively generates gauge fields from coarse to fine in link space, which we use as a baseline in our experiments. For an extended discussion, see \cref{app:related}.
\section{Proposed Method: Plaquette-Space Sampler (PSS)}
\label{sec:method}

% Since the target \eqref{eq:action} depends on the link field only through its plaquettes, we
% sample plaquettes. The difficulty, as discussed in \cref{sec:bg.lgt} is that the plaquettes of a
% link field are not independent, they must obey the Bianchi constraint. 
Our approach is to construct a generative sampler that generates valid plaquette configurations satisfying all Bianchi constraints exactly,
thereby ensuring the existence of corresponding link configurations.
Although a single-level plaquette sampler can in principle be constructed by solving all Bianchi constraints simultaneously (see \Cref{app:singlelevel}), it does not scale to large lattice sizes in higher-dimensional theories, as  
the resulting linear system grows rapidly, potentially posing computational challengs.
% \textcolor{red}{and the solving the linear could become challenging for large lattices} (see \cref{app:solve_size}).
Accordingly, we employ a multilevel sampling scheme.

% \subsection{The multilevel construction}
% \label{sec:m.refine}
%The central idea is to build 
% a valid plaquette field is built from coarse to fine. Throughout, \emph{valid} means that every cube identity \eqref{eq:bianchi_general} and every global closure holds, equivalently that the field is the plaquette field of some link field; the valid fields of the target lattice form the manifold $\mathcal M$ of \cref{sec:m.vars}. 

In a $d$-dimensional lattice with $V$ sites, the numbers of plaquette variables and independent Bianchi constraints are given by
$M = \frac{d(d-1)}{2}V $ and
$B = \frac{(d-1)(d-2)}{2}V + d-1$, respectively (see Appendix~\ref{app:dof}).
Let $\mathcal V^{(i)} (\subset G^{M^{(i)}})$ be the set of valid plaquette fields.
We first generate $d$ holonomies $\Omega \sim q_\Omega$ from their learned marginal distribution.
Then, at the coarsest level $i=0$, 
we generate $|F^{(0)}| = M^{(0)}-B^{(0)}$ free plaquette variables $ F^{(0)}
\sim \widetilde{q}_0(\cdot | \Omega)$ and determine the remaining $B^{(0)}$ plaquettes from the Bianchi identities (\cref{fig:refine}, left).
This yields the full plaquette configuation at this level,
$P^{(0)} = R_0(F^{(0)};\Omega)$,
where we refer to
 $R_0: G^{|F^{(0)}|} \times G^{d} \mapsto \mathcal V^{(0)}$
 as a \emph{constraint completion map}.
% We denote this coarsest map by $R_0: G^{\textcolor{red}{|F^{(0)}|}} \times G^{d} \mapsto \mathcal V^{(0)}$, which plays the role of $R_i$ below, with the holonomies in place of the coarse field, so that $P^{(0)} = R_0(F^{(0)};\Omega)$.

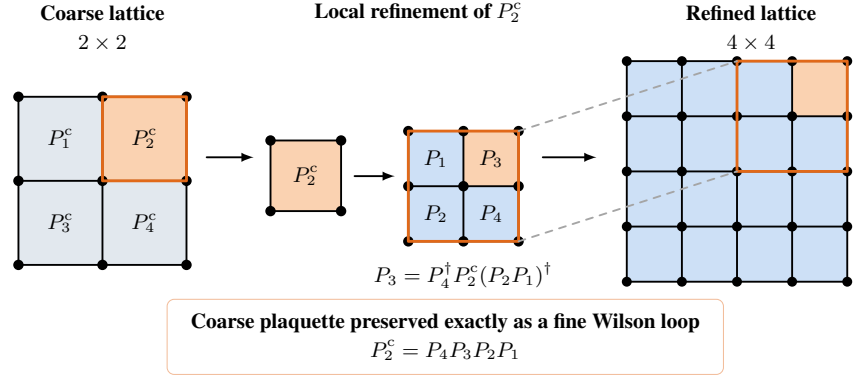
\begin{figure}[t]
\centering
% Reduce Figure 2 size by ~30%
\resizebox{0.8\linewidth}{!}{\begin{tikzpicture}[
    font=\small,
    >=latex,
    line cap=round,
    line join=round
]

\definecolor{coarsefill}{RGB}{226,232,238}
\definecolor{genfill}{RGB}{205,224,246}
\definecolor{detfill}{RGB}{250,210,177}
\definecolor{accent}{RGB}{225,105,25}
\definecolor{guide}{RGB}{165,165,165}

% ============================================================
% Coarse lattice: 2 x 2
% ============================================================
\node[font=\small\bfseries] at (1.5,5.10) {Coarse lattice};
\node[font=\footnotesize] at (1.5,4.65) {$2\times2$};

\begin{scope}[shift={(0.25,1.35)}]
    \def\s{1.25}

    \fill[coarsefill] (0,0) rectangle (\s,\s);
    \fill[coarsefill] (\s,0) rectangle (2*\s,\s);
    \fill[coarsefill] (0,\s) rectangle (\s,2*\s);
    \fill[detfill]    (\s,\s) rectangle (2*\s,2*\s);

    \draw[thick] (0,0) rectangle (2*\s,2*\s);
    \draw[thick] (\s,0) -- (\s,2*\s);
    \draw[thick] (0,\s) -- (2*\s,\s);

    \foreach \i in {0,1,2}{
        \foreach \j in {0,1,2}{
            \fill ({\i*\s},{\j*\s}) circle (2.2pt);
        }
    }

    \node at (0.5*\s,1.5*\s) {$P^{\mathrm c}_1$};
    \node at (1.5*\s,1.5*\s) {$P^{\mathrm c}_2$};
    \node at (0.5*\s,0.5*\s) {$P^{\mathrm c}_3$};
    \node at (1.5*\s,0.5*\s) {$P^{\mathrm c}_4$};

    \draw[accent, very thick]
        (\s,\s) rectangle (2*\s,2*\s);
\end{scope}

\draw[->, thick]
    (3.05,2.95) -- (3.75,2.95);

% ============================================================
% Local refinement of selected coarse plaquette
% ============================================================
\node[font=\small\bfseries] at (6.20,5.10)
    {Local refinement of $P^{\mathrm c}_2$};

\begin{scope}[shift={(4.00,2.15)}]
    \def\a{1.05}

    \fill[detfill] (0,0) rectangle (\a,\a);
    \draw[thick] (0,0) rectangle (\a,\a);

    \foreach \x/\y in {0/0,1/0,0/1,1/1}{
        \fill ({\x*\a},{\y*\a}) circle (2.2pt);
    }

    \node at (0.5*\a,0.5*\a) {$P^{\mathrm c}_2$};
\end{scope}

\draw[->, thick]
    (5.25,2.67) -- (5.85,2.67);

\begin{scope}[shift={(6.05,1.70)}]
    \def\f{0.82}

    \fill[genfill] (0,\f) rectangle (\f,2*\f);
    \fill[genfill] (0,0) rectangle (\f,\f);
    \fill[genfill] (\f,0) rectangle (2*\f,\f);
    \fill[detfill] (\f,\f) rectangle (2*\f,2*\f);

    \draw[thick] (0,0) rectangle (2*\f,2*\f);
    \draw[thick] (\f,0) -- (\f,2*\f);
    \draw[thick] (0,\f) -- (2*\f,\f);

    \foreach \i in {0,1,2}{
        \foreach \j in {0,1,2}{
            \fill ({\i*\f},{\j*\f}) circle (2.2pt);
        }
    }

    \node at (0.5*\f,1.5*\f) {$P_1$};
    \node at (0.5*\f,0.5*\f) {$P_2$};
    \node at (1.5*\f,1.5*\f) {$P_3$};
    \node at (1.5*\f,0.5*\f) {$P_4$};

    \draw[accent, very thick]
        (0,0) rectangle (2*\f,2*\f);
\end{scope}

\node at (6.87,1.20)
{
$\displaystyle
P_3
=
P_4^\dagger P^{\mathrm c}_2
(P_2P_1)^\dagger
$
};

% ============================================================
% Refined lattice: 4 x 4
% ============================================================
\draw[->, thick]
    (8.05,2.95) -- (8.85,2.95);

\node[font=\small\bfseries] at (11.15,5.10)
    {Refined lattice};
\node[font=\footnotesize] at (11.15,4.65)
    {$4\times4$};

\begin{scope}[shift={(9.30,1.10)}]
    \def\r{0.82}

    \foreach \i in {0,...,3}{
        \foreach \j in {0,...,3}{
            \fill[genfill]
            ({\i*\r},{\j*\r})
            rectangle
            ({(\i+1)*\r},{(\j+1)*\r});
        }
    }

    % determined plaquette in the refined image of P^c_2
    \fill[detfill]
        (3*\r,3*\r) rectangle (4*\r,4*\r);

    \draw[thick] (0,0) rectangle (4*\r,4*\r);

    \foreach \i in {1,2,3}{
        \draw[thick]
            ({\i*\r},0) -- ({\i*\r},4*\r);
        \draw[thick]
            (0,{\i*\r}) -- (4*\r,{\i*\r});
    }

    \foreach \i in {0,...,4}{
        \foreach \j in {0,...,4}{
            \fill ({\i*\r},{\j*\r}) circle (2.1pt);
        }
    }

    % fine Wilson loop corresponding to coarse P^c_2
    \draw[accent, very thick]
        (2*\r,2*\r) rectangle (4*\r,4*\r);
\end{scope}

% guide lines
\draw[guide, dashed, thick]
    (7.69,3.34) -- (10.94,4.38);

\draw[guide, dashed, thick]
    (7.69,1.70) -- (10.94,2.74);

% ============================================================
% Main message
% ============================================================
\node[
    draw=accent!45,
    rounded corners=3pt,
    inner xsep=10pt,
    inner ysep=6pt,
    align=center
] at (6.6,0.28)
{
\textbf{Coarse plaquette preserved exactly as a fine Wilson loop}
\\[2pt]
$\displaystyle
P^{\mathrm c}_2
=
P_4P_3P_2P_1
$
};

\end{tikzpicture}}
\caption{One refinement in plaquette space, for $d=2$. \textbf{Left:} the coarse lattice; three
plaquettes are free and $P^{\mathrm c}_2$ (orange) is fixed by the global closure.
\textbf{Centre:} its 
%$2\times2$ 
top-right block, refined
 with $P_1,P_2,P_4$ (blue) generated and $P_3$ determined
by $P_3=P_4^\dagger P^{\mathrm c}_2(P_2P_1)^\dagger$. \textbf{Right:} the refined lattice; the orange
outline marks the fine Wilson loop whose value is $P^{\mathrm c}_2$.}
\label{fig:refine}
\end{figure}

At each subsequent level, we generate the finer plaquettes $P^{(i)} \in \mathcal V^{(i)}$ conditional on the coarse plaquettes $P^{(i-1)}$, in addition to the holonomies $\Omega$.  Here, unlike the approaches of \citet{bauer2025super} and \citet{Abbott2023},
%whose fine-lattice flow acts on every link, 
we preserve the coarse plaquette variables by imposing local blocking constraints that fix the corresponding Wilson loops at the finer level to the same values (\cref{fig:refine}, center and right).
This makes the number of free plaquette variables generated at level $i$ be $|F^{(i)}| = \big(M^{(i)}-B^{(i)}\big) - \big(M^{(i-1)}-B^{(i-1)}\big)$.
% , the degrees of freedom of the valid fine field not already carried by the coarse one; for a doubling along one direction $N^{(i)}=(d-1)V^{(i-1)}$ (\cref{app:refine}).
%
% we generate new independent plaquettes and determine the remaining plaquettes exactly from local blocking and Bianchi identities, such that the coarse plaquette value is preserved as the value of the corresponding Wilson loop at the finer level.
 % If $P^{(i-1)}$ denotes a valid coarse field and $\mathcal V^{(i)}$ the set of valid fields on the next finer lattice, one refinement takes the form
% \begin{equation}
% F^{(i)}
% \sim q_i(\,\cdot\,|\,P^{(i-1)}),
% \qquad
% P^{(i)}
% =R_i(P^{(i-1)},F^{(i)})
% \in\mathcal V^{(i)},
% \label{eq:refine_ml}
% \end{equation}
% where $F^{(i)}$ contains the newly generated degrees of freedom and the fixed map $R_i$ completes the fine field. The learned model therefore only has to describe the conditional distribution of the new variables; satisfaction of the Bianchi constraints is handled exactly by the refinement map. Long-range structure is inherited from the coarse field, while the constraints introduced at each refinement remain local.
Let $R_i: G^{|F^{(i)}|} \times \mathcal V^{(i-1)} \times G^d  \mapsto \mathcal{V}^{(i)} $ 
be 
 the constraint completion map at level $i$. 
%from the free plaquette variables to the valid plaquette configurations.
Then, the sampling process is given as
\begin{equation}
F^{(i)}
\sim \widetilde{q}_i(\,\cdot\,|\,P^{(i-1)}, \Omega),
\qquad
P^{(i)}
=R_i(F^{(i)} ; P^{(i-1)}, \Omega).
%\in\mathcal V^{(i)}.
\label{eq:refine_ml}
\end{equation}
% where $F^{(i)}$ contains the newly generated degrees of freedom and the fixed map $R_i$ completes the fine field. The learned model therefore only has to describe the conditional distribution of the new variables; satisfaction of the Bianchi constraints is handled exactly by the refinement map. Long-range structure is inherited from the coarse field, while the constraints introduced at each refinement remain local.
% The learned model therefore only has to describe the conditional distribution of the new variables; satisfaction of the Bianchi constraints is handled exactly by the refinement map. Long-range structure is inherited from the coarse field, while the constraints introduced at each refinement remain local.
 Construction of the constraint completion map $R_i$ depends strongly on the lattice dimension and the gauge group, with explicit forms
detailed in \cref{app:refine}.
However, it is always fixed, bijective, and volume-preserving.
% one doubling of the lattice along a single direction, and the chain of levels it generates are given in \cref{app:refine}; the two-dimensional $2\times2$ block refinement of \cref{fig:refine} is in \cref{app:block2d}.

Let ${q}_i(\,\cdot\,|\,P^{(i-1)}, \Omega) = [R_i(\cdot; P^{(i-1)}, \Omega)]_{\#}
\widetilde{q}_i(\cdot |P^{(i-1)}, \Omega)$ be the push-forward of $\widetilde{q}_i$ under $R_i$.
Then, our complete \emph{multilevel Plaquette-Space Sampler} (PSS) is given as
%We model these coordinates level by level,
\begin{equation}
q(P, \Omega)
= q_{\Omega}(\Omega) \textstyle q_0(P^{(0)}| \Omega)
\prod_{i=1}^{n}
q_i(P^{(i)} |P^{(i-1)}, \Omega)
\label{eq:chain}
\end{equation}
with all densities with respect to the 
%product Haar measure 
 intrinsic measure on the Bianchi-constrained plaquette space.
% on the free coordinates $F^{(0)},\dots,F^{(n)}$, which the unit-Jacobian maps $R_i$ carry to the constraint set (\cref{prop:abelian})}
% Since the coordinate transformation \eqref{eq:coords} has unit Jacobian, this is a density on $\mathcal V$ and no additional Jacobian factor enters. The coarse variables therefore form a valid plaquette field themselves, while the fine variables describe only the new physical degrees of freedom. No separate coarse-lattice target distribution is introduced. The full hierarchy is trained against the target on the final lattice, so that at the optimum $q_0$ represents the coarse marginal of the fine target and each $q_i$ the corresponding conditional distribution.
We model all conditionals $\widetilde{q}_i$ by normalizing flows with the exact sampling density given by $q(P, \Omega)
= \textstyle 
 q_{\Omega}(\Omega) q_0(P^{(0)}| \Omega)
\prod_{i=1}^{n}
q_i(R^{-1}_i(P^{(i)}; P^{(i-1)}, \Omega) |P^{(i-1)}, \Omega)$.
We train the model by minimizing the reverse Kullback--Leibler divergence
\begin{equation}
\mathcal L(\theta)
=
\mathbb E_{P, \Omega \sim q_\theta}
\big[\log q_\theta(P, \Omega)+S(P)\big]
=
\KL(q_\theta\|p)-\log Z .
\label{eq:rkl}
\end{equation}
Architectures and training schedules are given in \cref{app:arch,app:training},
and the sampling procedure is summarized in \Cref{alg:sampler}.

\begin{algorithm}[tb]
\caption{Multilevel Plaquette-Space Sampler (PSS)}
\label{alg:sampler}
\begin{algorithmic}[1]
\STATE 
% $P^{(0)}\sim q_0$;\quad $\ell\leftarrow\log q_0\big(P^{(0)}\big)$
$\Omega \sim q_\Omega,
\quad
F^{(0)}\sim \widetilde{q}_0(\cdot|\Omega),   \quad P^{(0)}\leftarrow R_0(F^{(0)};\Omega),
\quad \ell\leftarrow \log 
q_\Omega(\Omega)
\widetilde{q}_0(F^{(0)}|\Omega)$
\hfill coarsest lattice
\FOR{$i=1,\dots,n$}
\STATE
$F^{(i)}\sim \widetilde{q}_i(\,\cdot\mid P^{(i-1)}, \Omega)$
\hfill learned conditional on $G^{|F^{(i)}|}$
\STATE
$P^{(i)}\leftarrow R_i(F^{(i)}; P^{(i-1)},\Omega)$
\hfill exact deterministic refinement
\STATE
$\ell\leftarrow\ell+\log \widetilde{q}_i(F^{(i)}\mid P^{(i-1)}, \Omega)$
\ENDFOR
\STATE
$P\leftarrow P^{(n)}$;\qquad
$S(P)$ from \eqref{eq:action}
\RETURN
$P$,\quad
$\log q(P, \Omega)=\ell$,\quad
$\log w=-S(P)-\ell$
\end{algorithmic}
\end{algorithm}

% {\color{red} Say advantages, novelty, etc. of our multilevel sampler, referring to experimental results}

The key advantage of the multilevel construction is that it satisfies the Bianchi identities exactly while keeping the solve local, and still retains a tractable sampling density. Rather than solving the Bianchi identities simultaneously,
we enforce them level by level: every generated configuration lies on the
constraint manifold by construction, while each determined plaquette depends
on at most four newly generated variables.
%, independently of the lattice volume. 
This avoids the scaling problem of the linear solve in higher-dimensional theories, as the matrix size %and condition number 
of the linear systems remains small, independent of the lattice size.
Because the constraint completion maps are bijective and Haar-volume preserving,
the resulting multilevel density can be evaluated exactly and used for
importance reweighting. At the same time, coarse plaquettes are preserved
exactly as Wilson loops of the refined field, allowing long-range structure to
be carried by the coarse levels while finer levels learn only the additional
 short-distance degrees of freedom. %\textcolor{red}{This is an important distinction from
% multiscale link-space flows \citep{Abbott2023}, which use the hierarchy as a
% prior for a separate flow on every fine link, unconstrained by the hierarchy: here the coarse
% variables are retained exactly through every refinement, and the only
% finest-lattice layers, four per-direction passes at four of the $4$D couplings,
% are refinements of the same constraint-preserving form (\cref{tab:arch},
% \cref{app:related}).} 
This is an important distinction from the
multiscale link-space flows of \citet{Abbott2023}, where the hierarcy is used only as a
prior, while an additional finest-lattice flow is responsible for capturing all correlations, thereby limiting the expressivity of the overall model.
% Here the coarse values are fixed exactly, so the long-range structure is inherited from the coarse levels and never has to be relearned at the finest resolution.
%and do not preserve the coarse-level link values as strict constraints on the final samples.
% \textcolor{red}{Their final flow acts on every fine link, unconstrained by the coarse values, so it is responsible for all correlations the hierarchy misses, long-range ones included, and must model them with fine flow on the full fine lattice, which is where link-space flows degrade toward weak coupling. Here the coarse values are fixed exactly, so the long-range structure is inherited from the coarse levels and never has to be relearned at the finest resolution.}
These properties are reflected in the experiments of
\cref{sec:exp}: the method scales more favourably with lattice size in
two-dimensional $U(1)$, 
shows an increasing advantage over link-space
baselines toward weak coupling in four-dimensional $U(1)$,
and remains efficient as correlations grow in
two-dimensional $SU(2)$.

%
% \subsection{Training}
% \label{sec:m.train}
%
% All models are trained without data 
% We train the mutlilevel sampler by minimizing the reverse Kullback--Leibler divergence
% \begin{equation}
% \mathcal L(\theta)
% =
% \mathbb E_{P\sim q_\theta}
% \big[\log q_\theta(P)+S(P)\big]
% =
% \KL(q_\theta\|p)-\log Z .
% \label{eq:rkl}
% \end{equation}
% where samples from $q_\theta$ are generated by \cref{alg:sampler}. Since every refinement map $R_i$ satisfies the constraints exactly, the objective is evaluated only on valid gauge fields, irrespective of the current quality of the learned conditionals.
%
% The gradient is obtained by differentiating through the full coarse-to-fine sampling chain. Thus all levels are trained jointly against the target distribution on the final lattice; no training data or separately defined coarse-lattice target distributions are required. Models are evaluated using the divergence, evidence recovery, and effective sample size defined in \eqref{eq:diagnostics}, together with reweighted observables compared against an independent reference (\cref{app:estimation}). 
% Architectures and training schedules are given in \cref{app:arch,app:training}.

\section{Numerical Experiments}
\label{sec:exp}

We evaluate our multilevel PSS in three complementary settings. 
Two-dimensional $U(1)$ provides a controlled test of scaling with lattice size
along a line of constant physics; four-dimensional
$U(1)$ probes its behaviour toward weak coupling, where link-space flows are known to
deteriorate; and two-dimensional $SU(2)$ tests whether the advantage
persists for a non-Abelian theory as the correlation length grows. 

\paragraph{Evaluation Measures.}
For $N$ generated samples $P_i \sim q$, we evaluate the KL divergence, the relative accuracy (RA), and the effective sample size (ESS): 
\begin{equation}
  \KL(q\|p)= \textstyle \log Z-\frac1N\sum_i \log w_i ,
  \qquad
  \mathrm{RA} =
  \hat{\mathcal{O}} / \mathcal{O}^{\mathrm{ref}} - 1,
  \qquad
  \ESS=\frac{(\sum_i w_i)^2}{N\sum_i w_i^2},
  \label{eq:diagnostics}
\end{equation}
where $ w_i=\exp(-S(P_i)) /  q(P_i)$ is an unnormalized importance weight
and $ \hat{\mathcal{O}} =\frac{\sum_i w_i \mathcal{O}(P_i)}{\sum_i w_i}$ is the self-normalized importance sampling estimator for the observable $\mathcal{O}$. Unless stated otherwise $\mathcal O$ is the \emph{mean plaquette}, the quantity summed by the action \eqref{eq:action}, averaged over lattice configurations. It is written $\langle\cos\varphi\rangle$ for $U(1)$ and $\langle\tfrac12\operatorname{tr}U_p\rangle$ for $SU(2)$, and is what panel~(b) of \cref{fig:scaling_te,fig:u1_4d,fig:su2_kl_obs_ess} reports.
The true free energy $\log Z$ and  reference values $\mathcal{O}^{\mathrm{ref}} $ for the observables are computed independently by sufficiently accurate methods (see \cref{app:reference,app:estimation}).
The KL divergence
$\KL(q\|p)$ measures the discrepancy between the sampler density $q$ and  the target density $p$, with the corresponding values for some baseline methods reported in their original papers.
RA allows us to assess both the bias and variance of the estimator, while 
$\ESS$ measures the sampling efficiency.
% dispersion of the importance weights, while $\log\hat Z$ measures how much of the target
% normalisation is recovered by the generated sample. The commonly used surrogate
% $\log\hat Z-\overline{\ell}$ underestimates the true divergence when important regions of
% the target are missed. We therefore also report
% $\exp(\log\hat Z-\log Z)$, which we call the \emph{evidence recovery}.
Crucially, $\ESS$ is reliable only when the estimator is consistent with the reference value, i.e., when RA shows a small bias.  Otherwise, $\ESS$ can be misleadingly large; for example, $\ESS=1$ when the sampler repeatedly generates the same configuration, regardless of the target distribution.

% \paragraph{Baseline Methods.}
% {\color{red} Can you summarize the baseline methods for each theory, and say which results are reproduced and which results are taken from the original paper?}
\paragraph{Baseline Methods.}
For each theory, we compare our multilevel PSS 
%with a single-level plaquette sampler and 
with link-space flow baselines. 
% \textcolor{red}{, exposing the effect of the hierarchy; in four dimensions the two are compared at matched parameter count, though they still differ in width, depth and seed (\cref{fig:ablation})}. 
For two-dimensional $U(1)$, we implement two link-space baselines ourselves: a
multiscale model \citep{Abbott2023} using the gauge-equivariant coupling
layers \citep{Kanwar_2020} with rational-quadratic splines, and a single-scale
version of the same flow applied directly to the full link lattice with a Haar
base.
For
four-dimensional $U(1)$, we compare our sampler with the multiscale and single-scale link-space samplers
of \citet{Abbott2023}. Since their code is not publicly available, we quote the
KL-divergence 
%central values 
values from Fig.~4 of \citet{Abbott2023}.
% ; uncertainties are shown there but
% could not be extracted reliably. These are the only numerical results taken
% directly from prior work.
 For two-dimensional $SU(2)$, we train a (single-level) continuous flow
\citep{Gerdes:2024rjk} using the released \texttt{bijx} library. Since neither training
script nor pretrained weights are provided for our benchmark, we performed training and
evaluation independently. 
% At the overlapping couplings, our
% reproduction gives $\ESS=0.85$ and $0.79$ at $\beta=2.2$ and $2.7$, compared with
% their reported $0.87$ and $0.68$.
Our implementation closely reproduces their reported results, yielding 
%At the overlapping couplings, our
%reproduction gives 
$\ESS=0.85$ and $0.79$ at $\beta=2.2$ and $2.7$, respectively, compared with
their reported $\ESS=0.87$ and $0.68$.
To our knowledge, no multilevel link-space sampler has been proposed for two-dimensional $SU(2)$.
As an ablation study, we also evaluate the single-level PSS, which  uses
the same representation as multilevel PSS but generate all free plaquettes with a single sampler and solves all Bianchi constraints simultaneously
(\cref{app:singlelevel}).

\paragraph{Evaluation Protocol.}
Each trained model is evaluated once from a retained checkpoint using a fixed random seed, with all reported quantities \eqref{eq:diagnostics}---KL divergence, RA, and ESS---%
%, evidence recovery, and reweighted observables---
computed from the same set of independent samples. 
% We use the diagnostics of
% \eqref{eq:diagnostics}, with $\log Z$ obtained independently of the model: exactly in
% two dimensions and by thermodynamic integration in four dimensions. Reweighted
% observables are likewise compared with independent references, using exact
% finite-volume results in two dimensions and a heatbath ensemble in four dimensions.
Uncertainties are estimated by delete-block jackknife with $100$ blocks over the same
samples.
%, including the uncertainty of the reference $\log Z$ in the KL error.
All plaquette-space samplers and link-space baselines are evaluated with $N=\Neval$, except the two-dimensional $SU(2)$ link-space baseline with $N\le10^4$, where the continuous flow
\citep{Gerdes:2024rjk} generation
is too costly---roughly one second per configuration. Smaller evaluation samples tend to bias
finite-sample ESS estimates upward, so these choices do not artificially favour our
plaquette-space models. Full estimator definitions and uncertainty procedures are
given in \cref{app:estimation}.

\begin{figure}[tbh]
\centering
\includegraphics[width=\linewidth]{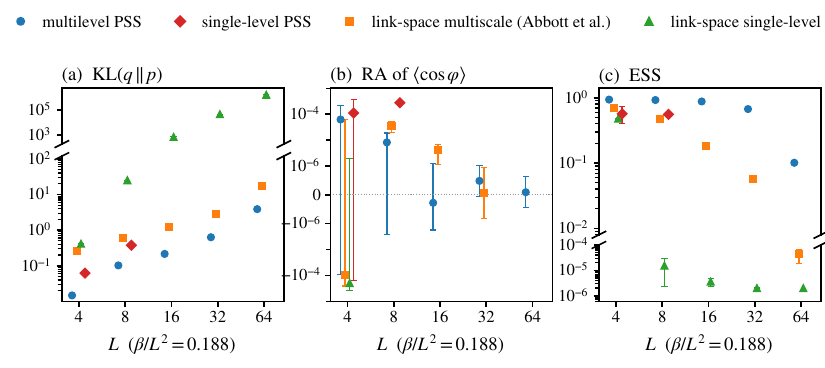}
\vspace{-12pt} % reduce gap between figure and caption
\caption{Two-dimensional $U(1)$ along the line of constant physics $\beta=\teC L^2$.
\textbf{(a)}~$\KL(q\|p)$. \textbf{(b)}~RA of the importance-weighted plaquette estimate against the exact finite-volume value.
\textbf{(c)}~$\ESS$. Panels (a) and (c) use logarithmic axes with one break;
points across a break are not directly comparable.}
\label{fig:scaling_te}
\end{figure}
\subsection{Results}

\paragraph{Two-dimensional $U(1)$.}
We begin with 
two-dimensional $U(1)$ theory, which
provides a direct test of scaling with lattice size. 
%In this case the refinement reduces to flux additivity within $2\times2$ blocks (\cref{app:refine}). 
Along a line of constant physics, $\beta=cL^2$, increasing $L$
keeps the physical theory fixed while resolving it on progressively finer lattices, so the
correlation length in lattice units grows. 
In \cref{fig:scaling_te} (a)-(c),
which show the three evaluation measures in Eq.\eqref{eq:diagnostics}, respectively,
both single-scale~\citep{Kanwar_2020} and multiscale~\citep{Abbott2023} link-space baselines deteriorate
rapidly with $L$ at $c=\teC$, yielding large KL divergences and large variances in observable estimates, as well as small $\ESS$ at $L=32$ and $64$.
% The single-scale flow has essentially one
% effective sample throughout, and the multiscale flow falls from
% $\ESS=\teMsEight$ at $L=8$ to $\teMsSixtyfour$ at $L=64$. 
In contrast, our multilevel
PSS 
%decreases from $\teOursEight$ to $\teOursSixtyfour$.
yields significantly lower KL divergence, compatible observable estimates, and larger ESS.
The advantage grows from $2.0\times$ at $L=8$ to $4.9\times$ at $L=16$ and
$11.9\times$ at $L=32$. At $L=64$, the weights of the link-space flow are severely degenerate: $\ESS=4\times10^{-5}$, with a single configuration carrying
$16\%$ of the total weight, compared with $\ESS=0.10$ for our multilevel PSS.
The multilevel PSS also reproduces the topological susceptibility $\chi_{\mathrm{top}}=\langle Q^2\rangle/V$ of the integer topological charge $Q$---a central physical quantity that is notoriously hard to estimate because Markov chains freeze in fixed-$Q$ sectors toward the continuum limit---within $0.7\sigma$ of its exact value for every $L$ (see 
%\cref{fig:chitop} in 
\cref{app:twod}).

\paragraph{Four-dimensional $U(1)$. }
\begin{figure}[tbh]
\centering
\includegraphics[width=\linewidth]{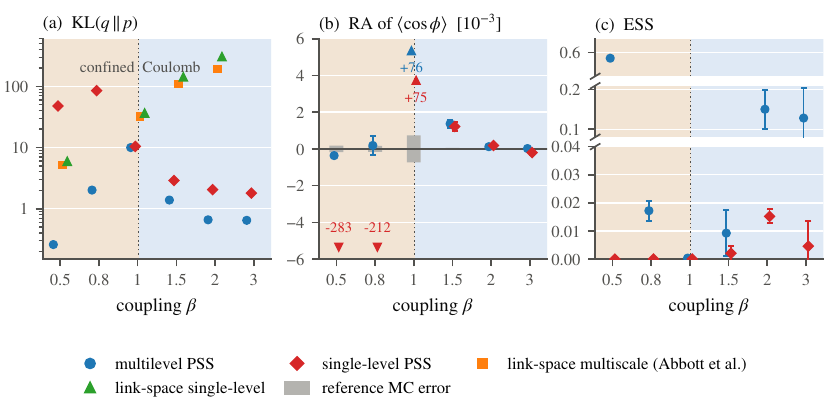}
\vspace{-12pt} % reduce gap between figure and caption
\caption{Four-dimensional $U(1)$.
\textbf{(a)}~$\KL(q\|p)$ as a function of $\beta$.
\textbf{(b)}~$\mathrm{RA}$ of the importance-weighted estimate of $\langle\cos\varphi\rangle$
against the heatbath reference, with the grey band indicating the reference's own uncertainty; points outside the plotted range are depicted as triangles at the panel edge, labeled with their values.
\textbf{(c)}~$\ESS$, shown on a linear scale with two axis breaks. Background shading denotes the confined and Coulomb phases, separated at
$\beta_c\simeq1.011$. }
\label{fig:u1_4d}
\end{figure}

\begin{figure}[tb]
\centering
\includegraphics[width=\linewidth]{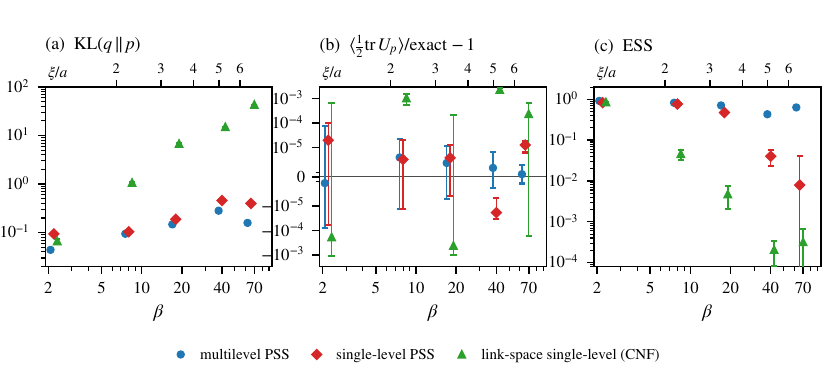}
\vspace{-12pt} % reduce gap between figure and caption
\caption{Two-dimensional $SU(2)$, $L=16$. \textbf{(a)}~$\KL(q\|p)$.
\textbf{(b)}~RA of the importance-weighted estiamte of $\langle\tfrac12\Tr U_p\rangle$ against the exact finite-volume
value. \textbf{(c)}~$\ESS$. 
% Three samplers: the multilevel and single-level plaquette
% flows and a link-space continuous flow. 
No multiscale link-space sampler has been proposed.
The upper axis indicates the correlation length
$\xi/a$.}
\label{fig:su2_kl_obs_ess}
\end{figure}

Next we test samplers
% in four-dimensional compact $U(1)$ 
on the $2^4\!\to\!4^4$ benchmark of
 \citet{Abbott2023}.
%  Our sampler contains five learned levels with eight coupling layers per
% % level.
% %and no additional fine-lattice flow. 
% The multiscale link-space baseline instead uses a
% hierarchical prior followed by $\AbbLayers$ gauge-equivariant layers on the fine lattice.
As shown in 
\Cref{fig:u1_4d}, our multilevel PSS achieves  smaller KL divergences than both the single-level and multilevel link-space baseline samplers.
Around the first-order transition at $\beta \simeq 1.011$, indicated by the background shading, all samplers perform poorly. Addressing this regime is beyond the scope of this work, as our multilevel PSS is designed to mitigate the issues that arise  % We instead focus on
in the weak-coupling regime at larger $\beta$, where the system is in the massless Coulomb phase and correlations are long-ranged, with their effective extent ultimately limited by the finite lattice size. This is precisely where the multiscale link-space construction deteriorates: its KL divergence increases with $\beta$, whereas our PSSs improve, leading to a rapidly widening performance gap. 
% At $\beta=1.5$ and $2$, our multilevel sampler achieves $\KL=1.4$ and $0.66$, respectively, compared with $110$ and $190$ for the multiscale link-space flow and $140$ and $300$ for its single-scale counterpart. 
As observed in \Cref{fig:u1_4d} (a), 
the gap in terms of the KL divergence  exceeds two orders of magnitude at larger $\beta$. 
\paragraph{Two-dimensional $SU(2)$.}
Two-dimensional $SU(2)$ tests the same construction in a non-Abelian theory as correlations become increasingly long-ranged. On the $16^2$ lattice, the correlation length $\xi/a$ (indicated by the scale along the top axis) grows from roughly one lattice spacing to $6.6$, so at the largest coupling it spans about $40\%$ of the lattice extent.
% Two-dimensional $SU(2)$ tests the same construction beyond the Abelian case. 
% % Each
% % refinement generates three group elements per $2\times2$ block and fixes the fourth through
% % the exact block relation (\cref{app:refine}). 
% On a $16^2$ lattice, the correlation length
% increases from roughly one lattice spacing to $\suXiMax$ over the range shown in
% \cref{fig:su2_kl_obs_ess}. 
The link-space flow is competitive only when correlations are short and rapidly loses efficiency as they grow. The single-level plaquette sampler is more robust, remaining close to the multilevel sampler at $\xi/a=2.2$, but it too deteriorates at longer correlation lengths and becomes effectively inefficient at $\xi/a=6.6$. Our multilevel PSS remains efficient throughout this regime and substantially outperforms both baselines.
% to our knowledge, no multilevel link-space sampler has been proposed for two-dimensional $SU(2)$.

% The link-space flow is competitive only when correlations are short and rapidly loses efficiency as correlation grows. The single-level plaquette sampler is substantially more robust, remaining close to the multilevel sampler at $2.2$, but it also deteriorates as correlations lengthen and becomes effectively inefficient at $6.6$.
% Our multilevel plaquette-space sampler substantially outperforms both (no multilevel link-space sampler has been proposed for 2-dimensional $SU(2)$).

% The single-level plaquette-space sampler also is competitive up to $\beta \sim 20$, 
% while deteriorates strongly as $\beta$ increases.
% The single-level sampler is competitive when correlations are short, with
% $\ESS=\suSLeight$ at $\beta=8$, but deteriorates strongly as $\beta$ increases, eventually
% reaching $\suSLlast$ and showing substantial sensitivity to the training seed. 
% The
% multilevel plaquette-space sampler remains much more stable: for $\beta\ge12$, its $\ESS$ stays within
% $\suMLrange$, yielding gains of up to $\suRatioMax$ over the single-level model.

Taken together, the three experiments show a consistent pattern. In two-dimensional $U(1)$ our multilevel PSS
scales substantially better with lattice size than the link-space baselines, while 
in four-dimensional
$U(1)$ and two-dimensional $SU(2)$, its relative advantage 
grows toward weak coupling and longer correlations. 
Although the single-level PSS performs comparably to the multilevel counterpart for small lattice sizes, its performance degrades as the lattice size increases, and scaling it further to larger lattices in higher-dimensional theories could become computationally challenging, as discussed in Appendix~\ref{app:solve_size}. 
\section{Conclusion }
\label{sec:conclusion}

We introduced a multilevel plaquette-space sampler (PSS) that satisfies the lattice
Bianchi identities exactly while preserving a tractable density. The construction
combines two complementary advantages: plaquette variables make the action local
and remove the local gauge redundancy, while the coarse-to-fine hierarchy keeps the remaining
constraint structure local---avoiding the scaling issues of solving the Bianchi-constrained problem---and carries long-range information across scales.

The numerical results show that this combination becomes increasingly beneficial
in the challenging regimes considered here: as we move toward the continuum limit in
two-dimensional $U(1)$, toward weak coupling in four-dimensional $U(1)$, and as
the correlation length increases in two-dimensional $SU(2)$. These results motivate extending
the method to larger four-dimensional lattices and non-Abelian gauge theories,
where exact local constraint handling and multilevel scale separation may offer
further advantages. More broadly, the same strategy may be useful for other
high-dimensional generative-sampling problems with structured constraints.

\paragraph{Limitations}
\label{app:limitations}
Our current experiments are restricted to relatively small lattices and to a single training run per model, and the extension to large four-dimensional non-Abelian systems remains to be demonstrated. In addition, although the Bianchi constraints are satisfied exactly by construction, sampling quality can still deteriorate in regimes with rare structures that are difficult for the continuous flow to represent. In particular, rare defect configurations, such as monopoles at weak coupling in four dimensions, may be under-sampled. Addressing these effects will require larger-scale studies and, potentially, more expressive models that explicitly capture such discrete or topological degrees of freedom.

\subsection*{AI use statement}

Generative AI tools were used to assist in language editing, literature discovery,
feedback on research methodology and experimental design, and aspects of code development,
organization, and finalization. They were not used to generate simulation data or report
numerical results. All simulations, training, and evaluation runs, and quantitative results
reported in this work, were produced by the authors using the procedures described in the paper.
All AI-assisted text, references, and code were reviewed and verified by the authors.
 Scientific ideas, methodological choices, implementation decisions, interpretation of
results, and final analysis remain the responsibility of the authors. The authors assume full
responsibility for the content of the article and the accompanying research artifacts.

 \subsubsection*{Acknowledgments}
% % Use unnumbered third level headings for the acknowledgments. All acknowledgments, including those to funding agencies, go at the end of the paper.
This work was supported by the German Federal Ministry of Education and Research (BMBF) under grant BIFOLD25B and by the European Union’s Horizon Europe Marie Sk\l{}odowska-Curie Doctoral Networks programme through the AQTIVATE project (grant agreement No.~101072344). 
This work is also supported with funds from the Ministry of Science, Research, and Culture of the State of Brandenburg within the Centre for Quantum Technologies and Applications (CQTA).
This project received funding from the European Research Council (ERC) via the project ”LEEX” grant agreement 101170304 funded by the European Union. Views and opinions expressed are however those of the author(s) only and do not necessarily reflect those of the European Union or the European Research Council Executive Agency (ERCEA). Neither the European Union nor the ERCEA can be held responsible for them.
We thank Stefan K\"uhn and Timo Eichhorn for helpful discussions and valuable comments on the manuscript.

\IfFileExists{refs.bib}{%
  \bibliographystyle{iclr2027_conference}\bibliography{refs}}{}

\appendix
% cleveref does not reliably rename section-level counters after \appendix: the version on
% Overleaf prints "Section E" where this one prints "Appendix E". These three lines force the
% appendix name on every cleveref version, and are a no-op when the fallback shim is active.
\providecommand{\crefalias}[2]{}
\crefalias{section}{appendix}
\crefalias{subsection}{appendix}
\crefalias{subsubsection}{appendix}
% Appendix order = the order the main text first reaches each topic:
% background (A-C), construction (D-H), evaluation (I-L), per-system results (M-O).
% File names carry their position; a file is renamed when it moves.

\section{Extended related work}
\label{app:related}

This appendix expands \cref{sec:bg.gen}: how flow samplers degrade with volume, the hierarchical
samplers our construction follows, the precise relation to the closest one, and the literature on
topological freezing that bounds what the method currently achieves.

\paragraph{Generative sampling for lattice field theory.}
Modern generative modelling includes several major families, including variational autoencoders
(VAEs) \citep{KingmaWelling2014VAE}, generative adversarial networks (GANs)
\citep{Goodfellow2014GAN}, normalizing flows \citep{RezendeMohamed2015NF,Dinh2017RealNVP}, and
diffusion models \citep{SohlDickstein2015Diffusion,Ho2020DDPM}. These have been adapted to
statistical physics and lattice field theory, where a generative model can represent a
high-dimensional configuration distribution and serve as an alternative sampling algorithm.
Flow-based sampling has received particular attention because its tractable density allows
generated configurations to be corrected by importance reweighting or Metropolis acceptance.
Normalizing flows for scalar lattice field theory were introduced by \citet{Albergo2019} and
extended to gauge theories \citep{Kanwar_2020,Boyda2021}, to stochastic flow constructions
\citep{Caselle:2022acb,Caselle:2022esc,Caselle:2024ent}, and to conditional sampling between
theory parameters \citep{Singha2023,Faraz:2023xdi,Singha:2023xxq,kanaujia2024advnf,sharma2026probing}. Normalizing
flows have also been used to reduce the variance of lattice QCD observables
\citep{abbott2026learning,Abbott:2026ylv} and to compute QED corrections in lattice field theory
\citep{hermansson2026normalizing}. Other paradigms have also been explored, several of them against
critical slowing down: GAN-based samplers reduce autocorrelation times in scalar theory
\citep{Pawlowski_2020} and address critical slowing in the Gross--Neveu model
\citep{10.21468/SciPostPhysCore.5.4.052}, variational autoencoders identify phase structure
\citep{Wetzel2017VAE}, and autoregressive and hierarchical networks give normalized samplers for
spin systems \citep{PhysRevLett.122.080602,Singha2025RiGCS,bialas2026sampling,bialas2026variational} and GMM based approaches~\cite{Faraz:2023xdi,sharma2026probing}.
Diffusion and related stochastic-path models have been applied to lattice field theory
\citep{Wang:2023exq,tan2026diffusion,chen2026stochastic,tomiya2026lattice} and extended to gauge theories
\citep{Zhu:2025pmw,aarts2026generalizableequivariantdiffusionmodels,Alharazin:2026lcb,Kanwar:2025wuc,Alharazin:2026lcb}; further
directions include physics-informed kernels
\citep{ihssen2025generativesamplingphysicsinformedkernels,ihssen2026solving} and Monte Carlo
estimates of flow fields \citep{Kanwar:2026ytz}.

\paragraph{Scaling with volume.}
Flow quality is known to deteriorate with volume \citep{Abbott:2022zsh}, and reverse-KL training
can collapse onto only part of a multimodal target \citep{Hackett:2021idh,Nicoli:2023qsl}. For
local theories the extensivity of $\KL$, together with sample-complexity bounds for importance
sampling \citep{chatterjee2018sample}, makes this particularly severe at large volume. Exact
reweighting remains possible whenever the proposal has sufficient support
\citep{nicoli_pre,nicoli_prl}, so the practical question is not correctness but the cost of an
accurate proposal. We address this through the representation and the constraint structure rather
than by introducing a new density model: the constraint layer of \cref{app:constraint} can be
combined with different tractable generative models, and we use spline couplings
\citep{Durkan2019} as a convenient choice for circular variables.

\paragraph{Hierarchical and multiscale samplers.}
Coarse-to-fine sampling has a long history in lattice field theory: multigrid Monte Carlo
\citep{Goodman:1986pv,Goodman:1989jw,GrabensteinPinn1994}, renormalization-group-guided
updates \citep{Schmidt:1983rng,FAAS1986571}, multiscale thermalization \citep{Endres:2015yca},
decimation maps \citep{Matsumoto:2023T0}, and multilevel variance reduction
\citep{PhysRevD.93.094507,Giles_2015}. Learned versions followed \citep{PhysRevLett.121.260601,PhysRevResearch.3.023230,PhysRevLett.129.136402,PhysRevLett.128.081603,Singha2025RiGCS,singha2026scalable,cotler2023renormalizingdiffusionmodels,masuki2026renormalization}. In lattice gauge theory,
renormalization-group schemes have been combined with normalizing flows for $U(1)$
\citep{Finkenrath:2024pR} and, in the construction we use as a baseline, for $SU(3)$
\citep{Abbott2023}. For discrete spin systems, hierarchical autoregressive networks
\citep{BIALAS2022108502} sample regions of a configuration in parallel with a shared
autoregressive network, following a recursive domain decomposition
\citep{PhysRevD.93.094507}, and have since been carried to three dimensions
\citep{Bialas:2025hxu}. Closest in spirit is RiGCS \citep{Singha2025RiGCS}, which builds a
configuration recursively from coarse to fine and conditions each level on the one below, using
conditional flows of the kind introduced by \citet{Singha2023} and carried to $U(1)$ gauge theory
by \citet{Singha:2023xxq}; related constructions extend the idea to continuous lattice field
theories as multilevel conditional flows \citep{singha2026scalable} and as super-resolution of
configurations \citep{bauer2025super,Efthymiou2018}. Our refinement levels are conditional in
exactly this sense. What is new is that the refined variables are \emph{constrained}, so every
level must additionally solve the Bianchi identity exactly, and the contribution is that this
solve can be kept local.

\paragraph{The closest precedent, and generative modelling under constraints.}
The closest precedent in lattice gauge theory is the multiscale flow of \citet{Abbott2023}, which
also refines one lattice direction at a time, but in \emph{link space}. There the hierarchy serves as a prior for a separate fine-lattice flow of
$\AbbLayers$ gauge-equivariant coupling layers on every link, which the authors state is required
because the hierarchy captures only a subset of the correlations; that flow is unconstrained by the
hierarchy and carries the final density. Here no such flow follows. Where a single-direction
refinement leaves some plaquette orientations last generated at a coarser resolution
(four-dimensional $U(1)$ at $\beta=0.5,0.8,2,3$), we add one further pass of the same refinement
per direction at the finest lattice: the field is re-parametrised as its blocking along $\mu$ and
the free variables of that refinement, only the free variables are transformed, conditioned on the
blocked field, and the constraint solve is unchanged (\cref{tab:arch}). Each pass fixes the
blocking along its own direction, and the density remains the product of exact conditionals
\eqref{eq:chain}. More
importantly, working in plaquette space introduces exact Bianchi constraints. The central issue is
therefore not anisotropic refinement itself, but how that refinement changes the structure of the
constraint solve: a global single-level parametrisation becomes nonlocal and increasingly ill
conditioned, while the multilevel construction keeps the solve local and its size and conditioning
independent of lattice extent (\cref{app:solve_size}). This notion of \emph{solve locality} is the
main distinction from existing hierarchical gauge-field samplers. It also places the work within
generative modelling under exact constraints. Gauge redundancy is commonly handled through
equivariance, in flows \citep{Kanwar_2020,Boyda2021} and in gauge-equivariant architectures
for lattice gauge theory more generally, such as lattice gauge equivariant convolutional networks
\citep{Favoni2022} and gauge-covariant networks for quarks and gluons \citep{Tomiya2021}; moving to
gauge-invariant variables removes that
redundancy but introduces equality constraints between the variables. Soft penalties do not
enforce these exactly, while projection generally changes the induced density and complicates
exact likelihood evaluation. We instead parametrise the constraint manifold directly, and its
two-dimensional realization (\cref{app:block2d}) shows that the mechanism is not specific to any
one dimension.

\paragraph{Topology and defects.}
The severe slowing down of topological modes toward the continuum limit
\citep{DelDebbio2004,Schaefer2011,Luscher2011} has motivated dedicated algorithms, including open
boundary conditions \citep{Luscher2011}, metadynamics \citep{Laio2016,eichhorn2021comparison,Eichhorn:2022NH}, tempering
\citep{Hasenbusch2017,Bonanno2021}, and winding updates \citep{Albandea:2021lvl}. Generative
approaches are more recent
\citep{Bonanno2026Flow,bonanno2026scalable,Rothkopf2026apb,Singha:2026aac} and show that global
learned proposals can substantially improve transitions between weakly connected sectors, whether
by resolving the sectors explicitly \citep{Singha:2026aac} or by letting the stochastic dynamics
jump between them \citep{Rothkopf2026apb}. In four-dimensional compact $U(1)$ the relevant defects
are magnetic monopoles together with the fluxes through the six two-tori. Our parametrisation
spans the full link configuration space and represents monopoles exactly through integer lifts of
the determined plaquettes (\cref{app:monopoles}); at weak coupling nonzero global flux sectors are
suppressed by $e^{-\fluxCost}$, so the trivial-sector heatbath reference reproduces the full
theory to exponential accuracy (\cref{app:reference}). Sampling such rare defect
configurations at the correct rate is a separate problem, beyond the scope of this work; the natural
route is to treat the discrete defect variables explicitly, in the same spirit that winding updates
introduce dedicated topology-changing moves for Markov chains. In two dimensions, where the sectors can be sampled directly, the multilevel construction
keeps the topological autocorrelation time near unity up to $L=32$, at a coupling where the
link-space baselines have already lost a resolved $\ESS$ (\cref{app:twod}).

\section{Degrees of freedom on \texorpdfstring{$T^d$}{T\^{}d}}
\label{app:dof}

\subsection{Tree gauge, holonomies and the link--plaquette pairing}
\label{app:dof.tree}

\paragraph{Tree gauge.}
Because of local gauge symmetry, the link variables contain redundant degrees of freedom. A
\emph{tree gauge} removes this redundancy by choosing a maximal spanning tree $\mathcal T$ of
$V-1$ links connecting all $V$ sites without forming a loop. Starting from a root, the gauge
transformation \eqref{eq:gauge} can be used successively to set every tree link to $\mathbf 1$.
This fixes the gauge everywhere up to the transformation at the root and leaves
$dV-(V-1)=(d-1)V+1$ independent links. In \cref{fig:dof} we use a comb-shaped tree: the bottom
row without its wrap-around link together with all vertical links outside the top row,
$(L-1)+L(L-1)=V-1$.

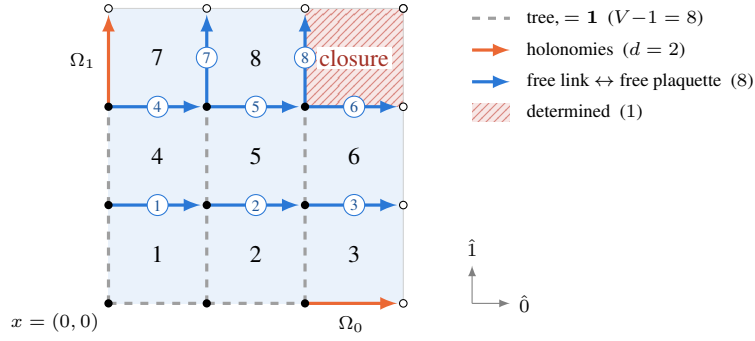
\begin{figure}[tbh]
\centering
% Degrees of freedom on a periodic 3x3 lattice in the COMB tree gauge of the 2D U(1) reconstruction
% (multilevel_u1_plq/plaquette_space/reconstruct.py): tree = bottom-row x-links except the wrap-around
% + all y-links except the top row; holonomies on the bottom wrap-around x-link (P_x) and the top y-link
% of column 0 (P_y); each remaining link is fixed by the plaquette carrying the same number, in that order;
% the corner plaquette is the global closure. Labels only.
\begin{tikzpicture}[font=\scriptsize, >=latex]
\definecolor{cfree}{HTML}{2A78D6}
\definecolor{chol}{HTML}{EB6834}
\definecolor{cdet}{HTML}{C0392B}
\def\s{1.3}
\tikzset{site/.style={circle, fill=black, inner sep=1.1pt},
         img/.style={circle, draw=black, fill=white, inner sep=1.0pt},
         tree/.style={gray!75, very thick, dashed},
         hol/.style={chol, very thick, ->},
         free/.style={cfree, very thick, ->},
         num/.style={circle, fill=white, draw=cfree, text=cfree!70!black, inner sep=1pt, font=\tiny},
         pnum/.style={font=\footnotesize}}
% plaquettes: (i,j) closed by the link with the same number; corner (2,2) = closure
\foreach \i/\j/\n in {0/0/1,1/0/2,2/0/3,0/1/4,1/1/5,2/1/6,0/2/7,1/2/8}{
  \fill[cfree!10] ({\i*\s},{\j*\s}) rectangle ({(\i+1)*\s},{(\j+1)*\s});
  \node[pnum] at ({(\i+0.5)*\s},{(\j+0.5)*\s}) {\n};}
\fill[cdet!12] ({2*\s},{2*\s}) rectangle ({3*\s},{3*\s});
\pattern[pattern=north east lines, pattern color=cdet!60] ({2*\s},{2*\s}) rectangle ({3*\s},{3*\s});
\node[pnum, text=cdet!85!black, fill=white, inner sep=1pt] at ({2.5*\s},{2.5*\s}) {closure};
% periodic images (faint)
\foreach \i in {0,1,2}{\draw[gray!35] ({\i*\s},{3*\s}) -- ({(\i+1)*\s},{3*\s});}
\foreach \j in {0,1,2}{\draw[gray!35] ({3*\s},{\j*\s}) -- ({3*\s},{(\j+1)*\s});}
% tree: bottom-row x-links i=0,1 ; all y-links with j=0,1
\foreach \i in {0,1}{\draw[tree] ({\i*\s},0) -- ({(\i+1)*\s},0);}
\foreach \i in {0,1,2}{\foreach \j in {0,1}{\draw[tree] ({\i*\s},{\j*\s}) -- ({\i*\s},{(\j+1)*\s});}}
% holonomies: bottom wrap-around x-link (P_x), top y-link of column 0 (P_y)
\draw[hol] ({2*\s},0) -- ({3*\s-0.08},0) node[midway, below=1pt, black] {$\Omega_0$};
\draw[hol] (0,{2*\s}) -- (0,{3*\s-0.08}) node[midway, left=1pt, black] {$\Omega_1$};
% free links: x-links of row 1 (1,2,3), row 2 (4,5,6); top y-links of columns 1,2 (7,8)
\foreach \i/\n in {0/1,1/2,2/3}{\draw[free] ({\i*\s},{1*\s}) -- ({(\i+1)*\s-0.08},{1*\s}); \node[num] at ({(\i+0.5)*\s},{1*\s}) {\n};}
\foreach \i/\n in {0/4,1/5,2/6}{\draw[free] ({\i*\s},{2*\s}) -- ({(\i+1)*\s-0.08},{2*\s}); \node[num] at ({(\i+0.5)*\s},{2*\s}) {\n};}
\foreach \i/\n in {1/7,2/8}{\draw[free] ({\i*\s},{2*\s}) -- ({\i*\s},{3*\s-0.08}); \node[num] at ({\i*\s},{2.5*\s}) {\n};}
% sites
\foreach \i in {0,1,2}{\foreach \j in {0,1,2}{\node[site] at ({\i*\s},{\j*\s}) {};}}
\foreach \k in {0,1,2,3}{\node[img] at ({3*\s},{\k*\s}) {}; \node[img] at ({\k*\s},{3*\s}) {};}
\node[below left=0pt] at (0,0) {$x=(0,0)$};
\draw[->, gray] ({3.7*\s},0) -- ({3.7*\s+0.5},0) node[right, black] {$\hat 0$};
\draw[->, gray] ({3.7*\s},0) -- ({3.7*\s},0.5) node[above, black] {$\hat 1$};
% legend
\begin{scope}[shift={({3.7*\s},{2.9*\s})}]
  \draw[tree] (0,0) -- (0.5,0); \node[anchor=west] at (0.6,0) {tree, $=\mathbf 1$ \ ($V{-}1=8$)};
  \draw[hol] (0,-0.42) -- (0.5,-0.42); \node[anchor=west] at (0.6,-0.42) {holonomies \ ($d=2$)};
  \draw[free] (0,-0.84) -- (0.5,-0.84); \node[anchor=west] at (0.6,-0.84) {free link $\leftrightarrow$ free plaquette \ ($8$)};
  \fill[cdet!12] (0,-1.37) rectangle (0.5,-1.13); \pattern[pattern=north east lines, pattern color=cdet!60] (0,-1.37) rectangle (0.5,-1.13);
  \node[anchor=west] at (0.6,-1.25) {determined \ ($1$)};
\end{scope}
\end{tikzpicture}
\caption{Degrees of freedom on a periodic $3\times3$ lattice in $d=2$ ($V=9$), shown in comb
gauge. Dashed grey links form the maximal tree and are fixed to $\mathbf 1$. Orange links are the
holonomies $\Omega_0,\Omega_1$. The eight numbered blue links pair one-to-one with the eight free
plaquettes, while the hatched corner plaquette is fixed by the global closure
\eqref{eq:global_closure_2d}. Hollow sites denote periodic images.}
\label{fig:dof}
\end{figure}

\paragraph{Holonomies.}
Among the remaining links, exactly $d$ close the non-contractible cycles of the torus, one in each
direction. Since every other link on such a cycle belongs to the tree and equals $\mathbf 1$, the
ordered product around the cycle reduces to the closing link itself: this is the Polyakov loop
$\Omega_\mu$. The residual gauge transformation at the root acts by conjugation,
$\Omega_\mu\mapsto g\,\Omega_\mu\,g^\dagger$, so its conjugacy class is gauge invariant.
Unlike plaquettes and their products, these loops are non-contractible and therefore cannot be
reconstructed from plaquettes alone.

\paragraph{Plaquettes and links.}
A plaquette sampler must distinguish the plaquettes that can be chosen freely from those fixed by
the lattice constraints. Of the $\binom d2 V$ plaquettes, the Bianchi identities for $d\ge3$ and
the global closures determine a subset $D$, leaving a free set $F$. The same counting appears from
the link side: after removing the $d$ holonomies from the $(d-1)V+1$ gauge-fixed links,
$(d-1)(V-1)$ degrees of freedom remain. Thus
\begin{equation}
  \underbrace{(d-1)V+1}_{\text{gauge-fixed links}}
  \;=\;\underbrace{d}_{\text{holonomies}}\;+\;\underbrace{(d-1)(V-1)}_{\text{free plaquettes}},
  \qquad
  \underbrace{\tbinom d2 V}_{\text{plaquettes}}
  \;=\;\underbrace{(d-1)(V-1)}_{\text{free}}\;+\;\underbrace{|D|}_{\text{determined}} .
  \label{eq:dof_count}
\end{equation}
This agreement is explicit in tree gauge: the $(d-1)(V-1)$ non-tree links pair one-to-one with the
free plaquettes. Adding them back sequentially, each closes a plaquette whose other links are already
known, so \eqref{eq:plq} determines the link from the plaquette, and vice versa. The remaining
plaquettes are then fixed by the Bianchi identities and global closures. Free plaquettes together
with the holonomies therefore provide coordinates for gauge-fixed link fields. For $d=2$, $V=9$,
$18=8+2+8$ links and $9=8+1$ plaquettes; for $d=4$, $V=256$,
$1024=255+4+765$ links and $1536=765+771$ plaquettes, giving a determined fraction of $50.2\%$.

The multilevel sampler does not choose its tree: the links each refinement sets to $\mathbf 1$
form one, and their union over levels is again maximal (\cref{app:refine.chain}). No level
traverses it, each closing its constraint in plaquette space. Should links be wanted --- for an
observable that winds around the lattice, and so is not a product of plaquettes --- the pairing
above recovers them for any group in any dimension, by running it forwards: choose a maximal tree,
draw the holonomies, and take each remaining link from the plaquette it closes.

\subsection{The same count without a gauge choice}
\label{app:dof.cochain}

The count also follows without a gauge choice, from the cochain complex of the periodic
lattice~\citep{Batrouni1982,BatrouniHalpern1984}. Let $C^k$ be the $k$-cochains --- sites, links
and plaquettes for $k=0,1,2$ --- and $\mathrm d:C^k\to C^{k+1}$ the coboundary, $\mathrm d^2=0$.
With $V$ sites in $d$ dimensions,
\begin{equation}
\dim C^k=\tbinom dk V ,
\qquad
b_k\equiv\dim H^k=\tbinom dk ,
\end{equation}
the second because $T^d$ carries $\binom dk$ independent $k$-cycles. Write $Z^k=\ker\mathrm d$ and
$B^k=\operatorname{im}\mathrm d$. The lattice is connected, so the constant $0$-cochains are the
one-dimensional kernel of $\mathrm d:C^0\to C^1$, giving $\dim B^1=V-1$ and, with $b_1=d$,
$\dim Z^1=(V-1)+d$. Rank--nullity for $\mathrm d:C^1\to C^2$ then gives
\begin{equation}
\dim B^2=dV-\dim Z^1=(d-1)(V-1),
\label{eq:dimB2}
\end{equation}
the number of independent continuous degrees of freedom in a link-generated plaquette field.

Local and global constraints enter separately. With $b_2=\binom d2$, so
$\dim Z^2=(d-1)(V-1)+\binom d2$, the local Bianchi identities restrict an arbitrary plaquette field
in $C^2$ to the closed subspace $Z^2$ and therefore number
\begin{equation}
M-\dim Z^2
=\tbinom d2 V-(d-1)(V-1)-\tbinom d2
=\Big[\tbinom d2-(d-1)\Big](V-1).
\label{eq:nlocal}
\end{equation}
The remaining $\binom d2$ directions of $Z^2/B^2\simeq H^2(T^d)$ are the fluxes through the
non-contractible two-tori, and the global closure relations select the link-generated fields $B^2$.
For compact $U(1)$ the integer flux sectors are distinct branches of those conditions; being
discrete labels, they leave \eqref{eq:dimB2} unchanged. Hence $|F|=\dim B^2=(d-1)(V-1)$,
$|D|=\binom d2V-(d-1)(V-1)$, and
\begin{equation}
f_d
=\frac{|D|}{\binom d2 V}
=1-\frac{(d-1)(V-1)}{\binom d2\,V}
\xrightarrow[V\to\infty]{}
1-\frac2d ,
\label{eq:frac_app}
\end{equation}
so $f_2=1/V\to0$ and $f_4=\tfrac12+\tfrac1{2V}\to\tfrac12$, as in \cref{sec:m.vars}.

On the link side, modding out the $V-1$ non-trivial gauge transformations leaves
$dV-(V-1)=\dim B^2+b_1$ continuous degrees of freedom: $\dim B^2$ carried by the plaquettes and
$d=b_1$ by the holonomies, which for $U(1)$ are flat directions of the Wilson action. Plaquettes
and holonomies therefore suffice to reconstruct a gauge-fixed link field, as \cref{app:dof.tree}
does explicitly.

The argument is linear, so it applies to $U(1)$ directly and otherwise componentwise in the Lie
algebra: for non-Abelian $G$ every continuous count gains a factor $\dim G$, while the Bianchi and
closure relations become non-linear and involve transport and the holonomies without changing the
local dimension count. On the finest four-dimensional lattice used here, $L=\Lfine$,
\begin{equation}
M=\nplq,
\qquad
\dim B^2=3(V-1)=3\cdot255=\dimBtwo,
\end{equation}
the number of free plaquette variables in every construction in this work.

\section{Parallel transport and local gauge frames}
\label{app:parallel_transport}

For a non-Abelian gauge theory, matrices based at different lattice sites are expressed in different
local gauge frames. Associate to each site $x$ a colour space $V_x\simeq\mathbb C^{N_c}$. A local
gauge transformation acts as an independent change of basis,
$\psi(x)\mapsto g(x)\psi(x)$, and a plaquette or Wilson loop based at $x$ transforms as
\begin{equation}
W(x)\mapsto g(x)W(x)g(x)^\dagger.
\label{eq:local_matrix_transform}
\end{equation}
A matrix $W(y)$ based at $y=x+\hmu$ transforms instead with $g(y)$. Hence the direct product
of matrices at different sites is not gauge covariant,
\begin{equation}
W(x)W(y)\mapsto
g(x)W(x)g(x)^\dagger g(y)W(y)g(y)^\dagger,
\end{equation}
because $g(x)^\dagger g(y)\neq\mathbf 1$ in general.

The link $U_\mu(x)$ identifies the neighbouring frames, transforming as \eqref{eq:gauge},
so a matrix based at $x+\hmu$ can be transported to the frame at $x$ as
\begin{equation}
W^{(x)}(x+\hmu)
\equiv
U_\mu(x)W(x+\hmu)U_\mu(x)^\dagger.
\label{eq:transport}
\end{equation}
Under a gauge transformation,
\begin{equation}
W^{(x)}(x+\hmu)
\mapsto
g(x)W^{(x)}(x+\hmu)g(x)^\dagger,
\label{eq:transported_transform}
\end{equation}
because the transformations at $x+\hmu$ cancel. Thus
$W(x)W^{(x)}(x+\hmu)$ is a well-defined covariant product in the common frame at $x$.

Therefore, before multiplying plaquettes with different base points---in particular in a
non-Abelian Bianchi identity---all factors must first be parallel transported to a common base
point. For $U(1)$, conjugation is trivial, $UWU^\dagger=W$, so this transport is invisible;
for $SU(2)$ and $SU(3)$ it must be treated explicitly.
\section{Lattice Bianchi identity}
\label{app:bianchi}

Consider an elementary cube based at $x$ and spanning the directions $\mu$, $\nu$ and $\rho$.
Its faces are the plaquettes \eqref{eq:plq}.
For a non-Abelian group, plaquettes based at different vertices cannot be multiplied directly,
so the three plaquettes on the far faces are transported back to the common base point $x$,
\begin{align}
\widetilde P_{\nu\rho}^{(\mu)}(x)
&=
U_\mu(x)\,
P_{\nu\rho}(x+\hmu)\,
U_\mu(x)^\dagger,
\nonumber\\
\widetilde P_{\mu\rho}^{(\nu)}(x)
&=
U_\nu(x)\,
P_{\mu\rho}(x+\hnu)\,
U_\nu(x)^\dagger,
\nonumber\\
\widetilde P_{\mu\nu}^{(\rho)}(x)
&=
U_\rho(x)\,
P_{\mu\nu}(x+\hrho)\,
U_\rho(x)^\dagger.
\label{eq:app_transported_plaq}
\end{align}
With a consistent orientation of the six faces, the ordered product around the boundary of the
cube is
\begin{equation}
\begin{split}
\mathcal B_{\mu\nu\rho}(x)
={}&
P_{\mu\rho}(x)\,
\widetilde P_{\mu\nu}^{(\rho)}(x)\,
P_{\nu\rho}(x)^\dagger
\\
&\times
\left[\widetilde P_{\mu\rho}^{(\nu)}(x)\right]^\dagger\,
P_{\mu\nu}(x)^\dagger\,
\widetilde P_{\nu\rho}^{(\mu)}(x).
\end{split}
\label{eq:app_bianchi_cube}
\end{equation}
The cancellation is visible in the first two factors already,
\begin{align}
P_{\mu\rho}(x)\,
\widetilde P_{\mu\nu}^{(\rho)}(x)
={}&
U_\mu(x)\,
U_\rho(x+\hmu)\,
U_\mu(x+\hrho)^\dagger\,
U_\rho(x)^\dagger
\nonumber\\
&\times
U_\rho(x)\,
U_\mu(x+\hrho)\,
P_{\mu\nu}(x+\hrho)\,
U_\mu(x+\hrho)^\dagger\,
U_\rho(x)^\dagger,
\end{align}
where $U_\rho(x)^\dagger U_\rho(x)$ and $U_\mu(x+\hrho)^\dagger U_\mu(x+\hrho)$ cancel. The
remaining faces are arranged in the same way by the transports and the surface ordering, so
that repeated use of $U^\dagger U=UU^\dagger=\mathbf 1$ removes every link and gives
\begin{equation}
\mathcal B_{\mu\nu\rho}(x)=\mathbf 1 .
\label{eq:app_bianchi_identity}
\end{equation}
No commutativity is used: transport and surface ordering arrange the factors so that each link
cancels against its inverse, which is why the identity holds for $SU(2)$ and $SU(3)$ exactly as
for $U(1)$.

\subsection{Which closed surfaces exist}
\label{app:bianchi.surfaces}

The identity \eqref{eq:bianchi_surface} holds for any closed surface $\Sigma$ assembled from
plaquettes: transporting its faces to a common base point and multiplying them in surface order
makes every link of $\Sigma$ occur once with each orientation, and the product collapses to the
identity. Which surfaces a periodic lattice carries therefore decides the form of the
constraint, and there are two kinds.

For $d\ge3$ the boundary of every elementary cube $c$ is a contractible closed surface, and each
one gives a \emph{local} relation: \eqref{eq:bianchi_surface} at $\Sigma=\partial c$,
\begin{equation}
  \textstyle \overrightarrow{\prod}_{p\subset\partial c}
  T_{p\to x_c}\!\left[P_p^{\epsilon_{cp}}\right]=\mathbf 1 ,
  \label{eq:bianchi_general}
\end{equation}
with $x_c$ any corner of the cube as the base point $x_0$ and
$\epsilon_{cp}\equiv\epsilon_{\partial c\,p}$ its face orientations. It is
\eqref{eq:app_bianchi_identity} written for a general cube.

In every dimension each coordinate plane $\mu\nu$ in addition carries the two-torus $\Lambda$ of
all its plaquettes, a closed but \emph{non-contractible} surface, and there the cancellation is
not complete. Cutting the torus along its two cycles opens it into a square whose boundary is
traversed as $a\,b\,a^{-1}b^{-1}$, with $a$ and $b$ the cycles in the $\mu$ and $\nu$ directions.
The plaquettes of $\Lambda$ tile that square, so transporting them to $x_0$ and multiplying in
surface order cancels every interior link as before, while the four sides are left over and
contribute the holonomies wrapping the plane in that order,
\begin{equation}
  \textstyle \overrightarrow{\prod}_{p\in\Lambda}
  T_{p\to x_0}[P_p]=\Omega_\mu\,\Omega_\nu\,\Omega_\mu^\dagger\,\Omega_\nu^\dagger ,
\end{equation}
one \emph{global} relation per plane, $\binom d2$ in all. This is the closure stated as
\eqref{eq:global_closure_2d} in the main text; it collapses to the identity only when the group
is Abelian, so that the two holonomies commute.

\subsection{Explicit forms for the three instantiations}
\label{app:bianchi.cases}

\paragraph{$U(1)$ in two dimensions.}
The links are the phases $U_\mu(x)=e^{i\theta_\mu(x)}$, so a plaquette is itself a phase,
$P_{\mu\nu}(x)=e^{i\varphi_{\mu\nu}(x)}$, whose angle is the oriented sum of its four link
angles,
\begin{equation}
  \varphi_{\mu\nu}(x)=\theta_\mu(x)+\theta_\nu(x+\hmu)-\theta_\mu(x+\hnu)-\theta_\nu(x).
  \label{eq:plq_angle}
\end{equation}
There are no cubes, so \eqref{eq:global_closure_2d} is the only relation; the group is Abelian,
the holonomies commute and the right-hand side is unity. Because the angles are defined modulo
$2\pi$, what remains is a single integer constraint on the whole lattice,
\begin{equation}
  \prod_{p\in\Lambda}P_p=1,
  \qquad
  \sum_{p\in\Lambda}\varphi_p=2\pi Q,
  \qquad Q\in\mathbb Z ,
  \label{eq:u1_global_closure}
\end{equation}
with the angles on their principal branch; the integer $Q$ is the topological charge.

\paragraph{$U(1)$ in four dimensions.}
The parametrisation is unchanged, but every elementary cube now contributes
\eqref{eq:bianchi_general}, which in the same way becomes an integer relation,
\begin{equation}
  \sum_{p\subset\partial c}\epsilon_{cp}\,\varphi_p=2\pi m_c,
  \qquad m_c\in\mathbb Z ,
  \label{eq:bianchi_u1}
\end{equation}
where $m_c$ is a monopole current, the four-dimensional counterpart of a magnetic charge. There
are in addition $\binom42=6$ global relations of the form \eqref{eq:u1_global_closure}, one per
coordinate plane. The difference from two dimensions is one of kind: a single global condition
there, an \emph{extensive} set of local ones here, one per cube. The integers $m_c$ and $Q$ are
not additional variables; they are read off from the angles, and they are what makes the
constraint discrete rather than merely geometric.

\paragraph{$SU(2)$ in two dimensions.}
A plaquette is a group element, written in the same form as the links,
$P_{\mu\nu}=\cos\omega_{\mu\nu}+i\sin\omega_{\mu\nu}\,\bm n\cdot\bm\sigma$, with $\bm n$ a unit
vector in $\mathbb R^3$ and the angle $\omega\in[0,\pi]$ fixed by
$\cos\omega=\tfrac12\operatorname{tr}P$. As in two-dimensional $U(1)$ there are no cubes and
\eqref{eq:global_closure_2d} is the only relation, but the group is non-Abelian: the parallel
transports and the surface ordering cannot be dropped, and the commutator on the right-hand side
does not reduce to unity, so the two holonomies enter the constraint itself rather than
cancelling out of it.

\section{Single-level plaquette sampling}
\label{app:singlelevel}

\subsection{Plaquettes as sampling variables}
\label{sec:m.vars}
The target distribution \eqref{eq:action} is constrained by the Bianchi identities and therefore
lies on a lower-dimensional manifold $\mathcal M$ of the full plaquette space. Since
$\mu(\mathcal M)=0$ in the ambient measure, a full-dimensional flow with density $q_\theta$ gives
$\Pr_{q_\theta}(P\in\mathcal M)=\int\mathbf 1_{\mathcal M}\,q_\theta\,d\mu=0$.
Penalty terms only suppress violations, while projection onto $\mathcal M$ generally destroys the
tractable density needed for likelihood evaluation and importance weighting. We therefore enforce
the constraints by construction, splitting the plaquettes into a free set $F$ and a determined set $D$.
Removing from the $\binom d2 V$ plaquettes the
$\big[\binom d2-(d-1)\big](V-1)$ independent cube identities
\eqref{eq:bianchi_general} and the $\binom d2$ global closures
\eqref{eq:global_closure_2d}, one per coordinate plane (\cref{app:dof}), leaves
\begin{equation}
|F|=(d-1)(V-1),
\qquad
f_d
=\frac{|D|}{\binom d2\,V}
=1-\frac{(d-1)(V-1)}{\binom d2\,V}
\xrightarrow[V\to\infty]{}
1-\frac{2}{d},
\label{eq:frac}
\end{equation}
so $f_2\to0$ while $f_4\to\tfrac12$. The free plaquettes together with the $d$ holonomies form
coordinates for gauge-fixed link fields (\cref{app:dof}, \cref{fig:dof}).

This count fixes how many plaquettes are determined, not how non-local the solve is. For $j\in D$,
let $k_j$ denote the number of free plaquettes on which it depends. Independent errors of size
$\epsilon$ then propagate as $O(\sqrt{k_j}\,\epsilon)$ in $j$, so Eq.~\eqref{eq:frac} alone gives no
information about the solve support $k_j$.

\subsection{The exact constraint layer}
\label{sec:m.build}
\label{app:constraint}

A flow generates the free plaquettes and a deterministic layer computes the determined ones,
provided the plaquettes are ordered so that every determined one is fixed by a relation among
plaquettes placed before it. The relations are the cube identities \eqref{eq:bianchi_general} in
$d\ge3$ and the single global closure \eqref{eq:global_closure_2d} in $d=2$. The layer holds to
machine precision on every sample, trained or not, and everything it needs is computed once,
offline: the orderings are properties of the lattice, not of the model or of $\beta$. The
multilevel construction reuses the same solve level by level, with the blocking relations in place
of the flux rows (\cref{app:blocking}).

\subsubsection{\texorpdfstring{$U(1)$}{U(1)}: an integer linear solve}
\label{app:constraint.u1}

\paragraph{The system.}
Let $\varphi\in\R^{M}$, $M=\binom d2V$, be the vector of all plaquette angles, unwrapped (in
$\R$ rather than $(-\pi,\pi]$). The relations satisfied by the plaquettes of a link field are
linear in $\varphi$ with integer coefficients, of two kinds. \emph{Bianchi rows}, one per
elementary cube $c$, the oriented sum of its six faces $\sum_{p\subset\partial c}\epsilon_{cp}
\varphi_p=0$, which is \eqref{eq:bianchi_u1} with $m_c=0$: $\binom d3V$ rows with entries
$0,\pm1$, of which $\big[\binom d2-(d-1)\big](V-1)$ are independent (\cref{app:dof.cochain}).
\emph{Flux rows}, one per plane $\mu\nu$,
\begin{equation}
  \sum_{a,b}\varphi_{\mu\nu}\big|_{x_\mu=a,\,x_\nu=b,\,x_\perp=0}=2\pi w_{\mu\nu},
  \label{eq:rowW}
\end{equation}
the sum over the plaquettes of one two-torus, the unwrapped form of
\eqref{eq:u1_global_closure}. Stacked,
\begin{equation}
  \M\varphi=b,\qquad
  \M=\begin{pmatrix}\M_{\mathrm B}\\ \M_{\mathrm W}\end{pmatrix},\qquad
  b=\begin{pmatrix}0\\ 2\pi w\end{pmatrix},
  \label{eq:system}
\end{equation}
with $\M_{\mathrm B}$ the Bianchi rows and $\M_{\mathrm W}$ the $\binom d2$ flux rows. In two
dimensions there are no cubes, $\M_{\mathrm B}$ is empty, and the system is the single flux row
$\sum_p\varphi_p=2\pi w$.

\paragraph{The solve.}
Split the columns into free and determined plaquettes, $F\sqcup D$, with
$|D|=\operatorname{rank}\M$ chosen so that the square block $\M_D$ is invertible. Then
$\M_F\varphi_F+\M_D\varphi_D=b$ gives the layer, an affine map on the unwrapped angles,
\begin{equation}
  \varphi_D=-\G\,\varphi_F+\M_D^{-1}b,
  \qquad
  \G=\M_D^{-1}\M_F\in\{-1,0,1\}^{|D|\times|F|},
  \label{eq:solve}
\end{equation}
in which $\G$ is the only object the model interacts with. In two dimensions $D$ is one plaquette
and $\G$ a row of ones, $\varphi_D=2\pi w-\sum_F\varphi_F$; in four dimensions it fixes half the
plaquettes from the other half. Setting $w=0$ loses nothing: the unwrapped plane sum of
$\mathrm d\theta$ telescopes to zero for every link field, so every configuration lies in that
parametrisation, and a non-zero physical flux enters only through the integer lifts below, costing
an action of order $\fluxCost$ where it could matter (\cref{app:reference}). Without the flux rows
the free space would contain the $\binom d2$ harmonic directions, which are the plaquette field of
no link configuration; with them its dimension is $(d-1)(V-1)$ (\cref{app:dof}), or $\dimBtwo$ at
$L=\Lfine$ in four dimensions.

\begin{proposition}[Constant Jacobian]
\label{prop:abelian}
For every $\varphi_F$ the point $\varphi=(\varphi_F,\varphi_D)$ satisfies \eqref{eq:system}
exactly, and its density on the constraint set with respect to $|F|$-dimensional Hausdorff
measure is
\begin{equation}
  \log q_{\mathcal M}(\varphi)=\log q_F(\varphi_F)-\tfrac12\log\det\!\big(\bm I_{|F|}+\G^{\!\top}\G\big),
  \label{eq:density}
\end{equation}
whose second term is independent of $\varphi_F$ and $b$.
\end{proposition}
\begin{proof}
The map $\varphi_F\mapsto\varphi$ is the affine embedding $\varphi=\varphi_0(b)+\bm B\varphi_F$
with $\bm B=(\bm I_{|F|},-\G^{\!\top})^{\!\top}$, whose volume element
$\sqrt{\det\bm B^{\!\top}\bm B}=\sqrt{\det(\bm I+\G^{\!\top}\G)}$ is a constant.
\end{proof}
The constant cancels in self-normalised importance weights and shifts the reverse-KL objective by
a constant, so it is never evaluated.

\begin{lemma}[Integrality]
\label{lem:torus}
$\varphi_D$ is well defined on $\T^{|D|}$ if and only if $\G$ is integer, which holds whenever
$\M_D$ is unimodular, i.e.\ every elimination pivot is $\pm1$.
\end{lemma}
\begin{proof}
Under $\varphi_k\mapsto\varphi_k+2\pi$ the solve shifts $\varphi_j$ by $-2\pi\G_{jk}$, a multiple
of $2\pi$ if and only if $\G_{jk}\in\Z$. All inputs are integer, so a non-integer can arise only
by division, and with every pivot of modulus one no division occurs. (Unimodularity is
sufficient, not necessary: $\M_D=\M_F=(2)$ gives $\G=(1)$.)
\end{proof}

\paragraph{Wrapped angles, monopoles and charge.}
The solve produces unwrapped angles. The physical plaquette is
$\bar\varphi_j=\operatorname{wrap}(\varphi_j)\in(-\pi,\pi]$, and the integer
$k_j=(\varphi_j-\bar\varphi_j)/2\pi$ is the \emph{lift} of the determined plaquette. Inserting the
wrapped angles into the relations returns \eqref{eq:bianchi_u1} and \eqref{eq:u1_global_closure}
with $m_c=-\sum_{p\subset\partial c}\epsilon_{cp}k_p$ and $Q=-\sum_p k_p$: the monopole numbers in
$d\ge3$ and the topological charge in $d=2$ are sums of lifts, read off the sample rather than
generated. The flow never sees an integer variable; the tails of its continuous density decide how
often a determined plaquette leaves $(-\pi,\pi]$, which is how monopoles arise in the
sampler (\cref{app:monopoles}).

\paragraph{Elimination and checks.}
The partition is computed in exact integer arithmetic, taking at each step a pivot from the row of
minimum remaining support (\cref{app:support}); the coarse $2^4$ system of the chain is reduced
with a fixed column priority. The elimination counts non-unit pivots rather than assuming there
are none, and in every system built for this work, at every level of every chain, the count is
zero and all coefficients of $\G$ lie in $\{-1,0,+1\}$. After elimination the rows of $\G$ are sparse
(\cref{tab:conditioning}: $O(1)$ non-zeros on a refinement level, a median support growing linearly with
$L$ for the single-level solve), so applying \eqref{eq:solve} is a sparse integer matrix--vector product. Whenever a system is built or a batch generated, three
conditions are asserted: every pivot is $\pm1$, the free-variable count equals $(d-1)(V-1)$, and
the configurations satisfy every relation to machine precision.

\subsubsection{\texorpdfstring{$SU(2)$}{SU(2)}: links in tree order}
\label{app:su2single}

For a non-Abelian group the relations are ordered products of plaquettes based at different sites,
so imposing them directly would mean deriving and solving an ordered identity for every relation.
The flow still generates plaquettes. The layer takes them, together with the sampled $d$
holonomies, and constructs the link field from which they come; the plaquettes of a link field
satisfy every relation automatically, so none need be imposed. No transporter appears in the
construction, which is what the tree buys: every site is reached from the root along identity
links, so plaquettes based at different sites may be multiplied at no cost. The lattice is
$L\times L$, $V=L^2$, and the pairing is that of \cref{app:dof.tree} and \cref{fig:dof}: in the
comb gauge every link outside the tree is the last unknown factor of exactly one plaquette, the
one it closes.

\paragraph{Generation.}
The flow produces $V-1$ elements $P_p\in SU(2)$, one for every plaquette except the corner
plaquette $p_\ast$ at $x=(L-1,L-1)$, which the global relation closes (\cref{fig:dof}). The two
holonomies $\Omega_0,\Omega_1$ are drawn from Haar measure.

\paragraph{Reconstruction.}
With the tree links equal to $\mathbf 1$ and the holonomies in place, the links follow in the
sweep order of \cref{app:dof}: each plaquette of row $j<L-1$ fixes the horizontal link above it,
and each plaquette of the top row fixes the top vertical link to its right. In both cases
\eqref{eq:plq} has one unknown factor and one solution; for the first, with
$P_{01}(x)=U_0(x)\,U_1(x+\hat0)\,U_0(x+\hat1)^\dagger\,U_1(x)^\dagger$,
\begin{equation}
  U_0(x+\hat1)=U_1(x)^\dagger\,P_{01}(x)^\dagger\,U_0(x)\,U_1(x+\hat0),
  \label{eq:su2_link_from_plaq}
\end{equation}
and analogously for the top row. One pass turns the $V-1$ generated plaquettes into the $V-1$ free
links, with no Jacobian: at every step the unknown link is the generated plaquette multiplied left
and right by elements already fixed, and Haar measure is invariant under both, so the density of
the link field relative to Haar equals the flow density on the $P_p$ times the Haar density of the
holonomies.

\paragraph{The determined plaquette.}
After the pass every link is known, so $P_{p_\ast}$ is read off them rather than solved for, which
is \eqref{eq:global_closure_2d} holding automatically:
\begin{equation}
  T_{p_\ast\to x_0}\big[P_{p_\ast}\big]
  =\Big(\overrightarrow{\prod_{p\neq p_\ast}}T_{p\to x_0}[P_p]\Big)^{\!\dagger}\,
   \Omega_0\,\Omega_1\,\Omega_0^\dagger\,\Omega_1^\dagger ,
  \label{eq:su2_closure_solve}
\end{equation}
the commutator of the holonomies divided by the transported product of all other plaquettes. It is
the single determined plaquette of \eqref{eq:frac} in two dimensions, with solve support $V-1$,
the non-Abelian counterpart of $\varphi_D=-\sum_F\varphi_F$. The action \eqref{eq:action} is
evaluated on all $V$ plaquettes, including $P_{p_\ast}$, and the importance weight is $w=e^{-S}/q$.

\paragraph{Higher dimensions.}
The same reconstruction applies in any $d$, with a nested comb in place of the two-dimensional one:
every link outside the tree is fixed by the plaquette it closes, and the plaquettes closing each cube
are then determined by \eqref{eq:bianchi_general}. The route is worth more there than in two
dimensions, the cube identities being $O(V)$ relations rather than one. We have not trained a
non-Abelian single-level model in $d\ge3$: there every level is a refinement of \cref{app:refine}
and no reconstruction is needed.

\subsection{Why one level is hard to learn}
\label{sec:m.global}

\begin{table}[tbh]
\centering
\caption{Solve support and conditioning of the $U(1)$ constraint layer. Single level: all
plaquettes of the lattice generated at once (\cref{sec:m.build}). Refinement level: one doubling
of the multilevel construction (\cref{app:refine}), whose values do not depend on the lattice
extent (\cref{app:support}).}
\label{tab:conditioning}
\vspace{4pt}
\setlength{\tabcolsep}{6pt}
\begin{tabular}{l cc c cc}
\toprule
 & \multicolumn{2}{c}{single level} & & \multicolumn{2}{c}{refinement level} \\
\cmidrule(lr){2-3}\cmidrule(lr){5-6}
lattice & support (median / max) & $\kappa(\bm I+\G^{\!\top}\G)$ & & support (max) & $\kappa$ \\
\midrule
$2^3$ & $5$ / $7$      & $20$   & & $4$ & $5$ \\
$8^3$ & $17$ / $159$   & $877$  & & $4$ & $5$ \\
$2^4$ & $5$ / $11$     & $\kappaSingleTwo$  & & $4$ & $7$ \\
$4^4$ & $\suppiso$ / $\suppisomax$ & $\kappaSingleFour$ & & $4$ & $7$ \\
\bottomrule
\end{tabular}
\end{table}
The layer is exact, but it changes the problem the flow has to solve. Solved on the whole lattice
at once, the constraint is global in the generated variables: a determined plaquette is fixed by a
chain of cube relations that ends on free plaquettes, and on a periodic lattice such chains cross
the lattice. \Cref{tab:conditioning} shows the consequence for $U(1)$, the solve support of a
single-level layer growing with the lattice and with it the conditioning of the target
(\cref{app:support}). At weak coupling $S\simeq\tfrac\beta2\|\varphi\|^2$, so substituting
\eqref{eq:solve},
\begin{equation}
  S\;\simeq\;\tfrac\beta2\,\varphi_F^{\!\top}\big(\bm I+\G^{\!\top}\G\big)\varphi_F+\text{linear},
  \label{eq:gauss}
\end{equation}
a Gaussian that couples every pair of free angles sharing a determined plaquette, with an
anisotropy given by $\kappa(\bm I+\G^{\!\top}\G)$. A coupling flow transforms a few variables per
layer and must build these long-range, badly scaled correlations implicitly; in two dimensions the
single determined plaquette depends on all $V-1$ others. A constraint that is local in the
physical variables has become global in the variables the model generates.

\section{The refinement and the chain}
\label{app:refine}

Notation follows \cref{sec:method}: $\mathcal V^{(i)}$ denotes the set of valid plaquette fields at
level $i$, $F^{(i)}$ the free variables generated at that level, $|F^{(i)}|$ in number, $R_i$ the refinement map
of \eqref{eq:refine_ml}, and $\Omega\in G^d$ the holonomies. Coarse quantities carry a superscript
$\mathrm c$: $P^{\mathrm c}$ for a coarse plaquette and $U^{\mathrm c}$ for a coarse link. All
statements below were verified numerically for $U(1)$ and $SU(2)$ in $d=2,\dots,5$
(\cref{app:verify}).

\subsection{Blocking, refinement and exactness}
\label{app:refine.defs}

\paragraph{The coordinates.}
A blocking map $\mathcal K_i:\mathcal V^{(i)}\to\mathcal V^{(i-1)}$, together with the free variables
$F^{(i)}\in G^{|F^{(i)}|}$, must define
\begin{equation}
\big(\mathcal K_i,F^{(i)}\big):\;
\mathcal V^{(i)}\longrightarrow\mathcal V^{(i-1)}\times G^{|F^{(i)}|}
\qquad\text{bijective, with unit Haar Jacobian,}
\label{eq:coords}
\end{equation}
with inverse $R_i$: a valid coarse field and arbitrary $F^{(i)}$ determine exactly one valid fine
field. We refine one direction $\mu$ at a time, in the gauge $U_\mu(2X)=\mathbf 1$, so that the two
fine plaquettes produced by a split are based at the same point $2X$. Then
\begin{equation}
\mathcal K_i:\qquad
P^{\mathrm c}_{\mu\nu}(X)=P_{\mu\nu}(2X+\hmu)\,P_{\mu\nu}(2X),
\qquad
P^{\mathrm c}_{\nu\rho}(X)=P_{\nu\rho}(2X),
\qquad (\nu,\rho\neq\mu),
\label{eq:block}
\end{equation}
and
\begin{equation}
R_i:\qquad
\begin{aligned}
P_{\mu\nu}(2X)&=F_\nu, &
P_{\mu\nu}(2X+\hmu)&=P^{\mathrm c}_{\mu\nu}F_\nu^\dagger,\\[2mm]
P_{\nu\rho}(2X)&=P^{\mathrm c}_{\nu\rho}, &
P_{\nu\rho}(2X+\hmu)&=F_\nu\,
T_\nu\big[F_\rho(X+\hnu)\big]\,
P^{\mathrm c}_{\nu\rho}\,
T_\rho\big[F_\nu(X+\hrho)\big]^\dagger F_\rho^\dagger ,
\end{aligned}
\label{eq:classes}
\end{equation}
where $T_\nu,T_\rho$ denote transport by coarse links and are therefore functions of
$P^{\mathrm c}$ and $\Omega$.

\paragraph{The four classes and their supports.}
\label{app:refine.classes}
The four lines of \eqref{eq:classes} exhaust the possible fine plaquettes: generated, fixed by its
blocking partner, copied, or determined by one cube identity. Their solve supports are respectively
$1,1,0,4$, independent of $G$, $d$, and the lattice extent. The first three follow directly from
\eqref{eq:refine}; the fourth is the cube identity \eqref{eq:bianchi_general} at $2X$ in the
directions $\mu,\nu,\rho$. Its two faces transverse to $\mu$ are $P_{\nu\rho}(2X)$ and
$P_{\nu\rho}(2X+\hmu)$, while the four faces containing $\mu$ are generated elements at
$\mu$-even sites. There are $d-1$ plaquettes of each class per coarse site, hence
\begin{equation}
  |F^{(i)}|=(d-1)V^{(i-1)},
  \label{eq:levelcount}
\end{equation}
which is the level count used in \cref{sec:method}. In $d=2$ there is no third direction and hence
no cube class; only the first line of \eqref{eq:classes} remains. For $U(1)$ the transporters drop
out and \eqref{eq:block} becomes linear in the angles,
\begin{equation}
\Phi_{\mu\nu}(X)=\varphi_{\mu\nu}(2X)+\varphi_{\mu\nu}(2X+\hmu),
\qquad
\Phi_{\nu\rho}(X)=\varphi_{\nu\rho}(2X),
\label{eq:rowK}
\end{equation}
with $\Phi$ coarse and $\varphi$ fine. In these coordinates, the condition number of
$\bm I+\G^{\!\top}\G$ is $1,5,7,9$ for $d=2,3,4,5$, respectively, at every extent tested
(\cref{app:support}, \cref{tab:verify}).

\paragraph{Exactness.}
\label{app:refine.exact}
We prove \cref{eq:coords} at the link level. Decimation along $\mu$ and its inverse are
\begin{equation}
  U^{\mathrm c}_\mu(X)=U_\mu(2X)\,U_\mu(2X+\hmu),
  \qquad
  U^{\mathrm c}_\nu(X)=U_\nu(2X)\quad(\nu\neq\mu),
  \label{eq:decimate}
\end{equation}
and
\begin{equation}
  U_\mu(2X)=\mathbf 1,\quad
  U_\mu(2X+\hmu)=U^{\mathrm c}_\mu(X),\quad
  U_\nu(2X)=U^{\mathrm c}_\nu(X),\quad
  U_\nu(2X+\hmu)=F_\nu(X)\,U^{\mathrm c}_\nu(X),
  \label{eq:refine}
\end{equation}
respectively. With $U_\mu(2X)=\mathbf 1$, \eqref{eq:decimate} reproduces the blocked plaquettes
in \eqref{eq:block}. The bijection and unit Jacobian then follow in three steps.

\emph{(i)} Substituting \eqref{eq:refine} into \eqref{eq:decimate} recovers $U^{\mathrm c}$, while
$P_{\mu\nu}(2X)=F_\nu(X)$ because the first and third factors of that plaquette are identities.

\emph{(ii)} Any $U'$ that decimates to $U^{\mathrm c}$ can be brought to the form
\eqref{eq:refine} by the gauge transformation \eqref{eq:gauge} with
$g(2X)=\mathbf 1$ and $g(2X+\hmu)=U'_\mu(2X)$. This leaves every decimated product unchanged and
is the unique transformation supported on the new sites with this property, since
$g(2X+\hmu)$ is fixed by
$g(2X)U'_\mu(2X)g(2X+\hmu)^\dagger=\mathbf 1$. The gauge choice is therefore necessary rather
than merely convenient: without it, one pair $(U^{\mathrm c},F)$ would correspond to an entire
orbit of fine fields. For $U(1)$ the plaquettes are gauge invariant, so no such choice is required.

\emph{(iii)} For fixed $U^{\mathrm c}$, the only links that depend on $F$ are
$U_\nu(2X+\hmu)=F_\nu(X)U^{\mathrm c}_\nu(X)$, one for each generated element. Each is a right
translation by an element independent of $F$. Since Haar measure is right invariant, product Haar
measure on the $F_\nu(X)$ is mapped to Haar measure on the fine links with unit Jacobian
(numerically $|\log|\det J||\le2\times10^{-14}$, \cref{tab:verify}).

\subsection{The chain}
\label{app:refine.chain}

We start from the one-site torus, whose $d$ links are the holonomies $\Omega$, and refine successively
along $\mu=0,\dots,d-1$, repeating this cycle until the target extent is reached. The links fixed to $\mathbf 1$ form a
maximal tree at every level. This follows by induction: the one-site torus fixes none, while a
refinement along $\mu$ fixes the $V^{(i-1)}$ links $U_\mu(2X)=\mathbf 1$. Each previously fixed
coarse link remains fixed after refinement, appearing either as $U_\nu(2X)$ or as
$U_\mu(2X+\hmu)$, so the total becomes
$(V^{(i-1)}-1)+V^{(i-1)}=V^{(i)}-1$. These links contain no loop, since any such loop would project
to a loop in the coarse tree. Summing \eqref{eq:levelcount} over the volumes
$1,2,4,\dots,V/2$ gives
\begin{equation}
  d+\sum_i |F^{(i)}|=d+(d-1)\big(1+2+\dots+\tfrac V2\big)=d+(d-1)(V-1),
  \label{eq:tree_count}
\end{equation}
which is precisely the count in \cref{sec:m.vars}: the holonomies together with the generated
elements give all degrees of freedom of a link field in tree gauge. Applying step (iii) at each level,
the full map to non-tree links has unit Haar Jacobian, so the density factorises as \eqref{eq:chain}
and the importance weight is $e^{-S}/q$.

\subsection{Plaquette form for \texorpdfstring{$U(1)$}{U(1)}}
\label{app:blocking}

For $U(1)$ the refinement can be written directly in terms of plaquette angles. Equation
\eqref{eq:rowK} is the plaquette-space image of \eqref{eq:decimate}: after substitution into
\eqref{eq:plq}, the interior link of a face containing $\mu$ appears twice with opposite signs, so
the coarse face is the sum of the two fine faces it covers, while a face not containing $\mu$ is
simply the fine face at even $x_\mu$. Since the blocking map is a sum followed by restriction, and
both commute with the coboundary, it intertwines the fine and coarse coboundaries,
\begin{equation}
  \mathrm d_{\mathrm c}\,\M_{\mathrm K}=\M_{\mathrm K}'\,\mathrm d_{\mathrm f}
  \quad\Longrightarrow\quad
  \mathrm d_{\mathrm f}\varphi=0\;\Rightarrow\;\mathrm d_{\mathrm c}\Phi=0 ,
  \label{eq:intertwine}
\end{equation}
so the coarse Bianchi identity is inherited rather than imposed again. The level system is
\begin{equation}
  \M\varphi=b,\qquad
  \M=\begin{pmatrix}\M_{\mathrm B}\\ \M_{\mathrm K}\end{pmatrix},\qquad
  b=\begin{pmatrix}0\\ \Phi\end{pmatrix},
  \qquad
  (\M_{\mathrm K}\varphi)_P=\sum_{p\in P}\varphi_p ,
  \label{eq:levelsystem}
\end{equation}
where $\M_{\mathrm B}$ contains the fine Bianchi rows and $\M_{\mathrm K}$ one blocking row per
coarse plaquette $P$. The coarse field enters only through the constant $b$; no flux rows are needed,
because the fluxes are fixed at the coarsest level and inherited through \eqref{eq:intertwine}. The
elimination and the two properties of \cref{app:constraint} then apply unchanged, and
\cref{app:support} compares these coordinates with those of \eqref{eq:classes}.

The four-dimensional chain combines the first four refinements, $1^4\to2^4$, into a single joint
flow on the $45=3\cdot16-3$ angles of the $2^4$ lattice, with fluxes fixed to zero. It then refines
one direction at a time, generating $48,96,192,384$ angles, copying the same numbers at support zero,
and solving $96,192,384,768$. Each step is a unimodular integer change of variables from the
$F_\nu(X)$ of \eqref{eq:refine}, so its Jacobian is one (\cref{tab:verify}). Refining all directions
at once is also the image of a link decimation and therefore also preserves the Bianchi identity, but
within a coarse cell it leaves fine faces that are not fixed by any single cube relation. The
constraints then chain across the cell and the support grows with the extent (\cref{app:support}).
Thus refining one direction at a time is not required for exactness, but for locality.

\subsection{Two-dimensional refinements}
\label{app:block2d}

In two dimensions both directions are refined at once, equivalent to two successive refinements along
$\mu=0$ and then $\mu=1$. Each $2\times2$ block is generated from its coarse plaquette using three
elements, while the fourth is fixed by
\begin{equation}
P^{\mathrm c}(X)=P_4P_3P_2P_1 ,
\label{eq:block2d}
\end{equation}
or equivalently
\begin{equation}
P_3=P_4^\dagger\,P^{\mathrm c}\,(P_2P_1)^\dagger .
\label{eq:refine2d}
\end{equation}
Hence $|F^{(i)}|=3V^{(i-1)}$, which is \eqref{eq:levelcount} applied twice. For $U(1)$ this is simply
\eqref{eq:rowK} applied twice in angle variables, and the layer is the $d=2$ case of
\cref{app:constraint}: one blocking row per block and no cube rows. The holonomies do not enter the
action and are drawn uniformly; the flux is fixed to zero at the coarsest level, and the topological
charge of a sample is the sum of the lifts of its determined plaquettes (\cref{app:constraint}).

For $SU(2)$ only the ordering changes, while \eqref{eq:block2d} remains unchanged. Refining first
along direction $0$ generates $F(X)$ at each coarse site and produces the intermediate plaquettes
$F(X)$ and $P^{\mathrm c}(X)F(X)^\dagger$. Refining along direction $1$ then generates one element
at each intermediate site. Labelling the block top-left, bottom-left, top-right, bottom-right
(\cref{fig:refine}),
\begin{equation}
  \begin{aligned}
  P_1&=G_L^\dagger F, & P_2&=G_L,\\
  P_3&=G_R^\dagger\,P^{\mathrm c}F^\dagger, & P_4&=G_R,
  \end{aligned}
  \qquad F,G_L,G_R\in SU(2)\ \text{generated},
  \label{eq:su2_block}
\end{equation}
where $G_L,G_R$ are the second-refinement elements at the two intermediate sites. The three links
$U_0(2X)$, $U_1(2X)$ and $U_1(2X+\hat0)$ inside the block are identities, so all four fine
plaquettes are based at $2X$ and no transporter appears. Generating $P_1,P_2,P_4$ and determining
$P_3$ through \eqref{eq:refine2d} gives
$P_3=G_R^\dagger P^{\mathrm c}F^\dagger$ with $F=P_2P_1$. The two parametrisations therefore
differ only by translations and inversions and have the same unit Haar Jacobian.

The implementation contains one transporter as a gauge artefact. The two-dimensional $SU(2)$ code
inherits the comb gauge of \cref{app:dof.tree} from the coarsest level, where the link
$a=U_0(2X)$ at the lower-left corner of a block need not be the identity. The flow therefore works
with the transported variables $Q_4=aP_4a^\dagger$ and $Q_3=aP_3a^\dagger$, for which
\eqref{eq:refine2d} again holds. This causes no complication: $a$ depends only on $P_1$ and $P_2$,
which are always generated, so no determined plaquette enters a transporter and all blocks can be
solved in parallel. Conjugation preserves both Haar measure and the trace, and the transformation
between the two descriptions is the identity at every block corner, leaving the coarse links unchanged
and allowing the levels to compose.

The holonomies are those of the coarsest level. They are drawn from Haar measure and remain unchanged
under blocking, since coarse links are products of fine ones. They enter the fine plaquettes only
through the coarsest plaquette, the commutator
$\Omega_0\Omega_1\Omega_0^\dagger\Omega_1^\dagger$, which is supplied to the conditioner at every
level (\cref{app:arch}). Their density is Haar rather than learned: the $n$-th character of the exact
marginal is suppressed as
$\exp\!\big[-(n^2-1)(L/\xi)^2/3\big]$, giving a deviation of $8.5\times10^{-3}$ at
$(L/\xi)^2=5$ and $2.3\times10^{-9}$ at $20$, for $(L/\xi)^2$ between $5.9$ and $196$ over the couplings studied.
% tab:support and fig:ablation stand before the heading on purpose: LaTeX places a [t] float at
% the top of a page only if it was met on the page before, so this puts both at the head of the
% page on which the section starts, with the text that discusses them.
\begin{table}[!t]
\centering
\caption{Support of the constraint solve and condition number $\kappa$ of $\bm I+\G^{\!\top}\G$
for the systems used in the four-dimensional experiments (integer elimination, minimum-support
pivoting). Support $0$ means the plaquette is copied from the coarse field. These are the coordinates
the trained models use. In the refinement coordinates of \cref{app:refine}, every determined plaquette
depends on at most four generated variables and $\kappa=\kappaRefFour$ (\cref{tab:dimscan}); the two
choices are related by a unimodular integer change of variables, and both are independent of the
lattice size.}
\label{tab:support}
\setlength{\tabcolsep}{3pt}
\begin{tabular}{lrrrrrrr}
\toprule
system (fine shape, direction) & plaquettes & generated & determined
   & median & mean & max & $\kappa$ \\
\midrule
coarse level $2^{4}$                                & 96   & 45  & 51   & 5  & 6.0  & 11 & 45 \\
$2^{4}\to 4\!\times\!2^{3}$, $\mu=0$               & 192  & 48  & 144  & 3  & 2.8  & 8  & 45 \\
$\phantom{2^{4}}\to 4^{2}\!\times\!2^{2}$, $\mu=1$ & 384  & 96  & 288  & 2  & 2.7  & 9  & 43 \\
$\phantom{2^{4}}\to 4^{3}\!\times\!2$, $\mu=2$     & 768  & 192 & 576  & 1  & 2.3  & 8  & 30 \\
$\phantom{2^{4}}\to 4^{4}$, $\mu=3$                & 1536 & 384 & 1152 & 1  & 1.7  & 4  & 7  \\
\midrule
isotropic step $2^{4}\to 4^{4}$                     & 1536 & 720 & 816  & 13 & 13.7 & 47 & --- \\
single level $4^{4}$ (\cref{fig:ablation})          & 1536 & 765 & 771  & 13 & 13.6 & 47 & 391 \\
\bottomrule
\end{tabular}
\end{table}

\begin{figure}[!t]
\centering
\includegraphics[width=0.55\linewidth]{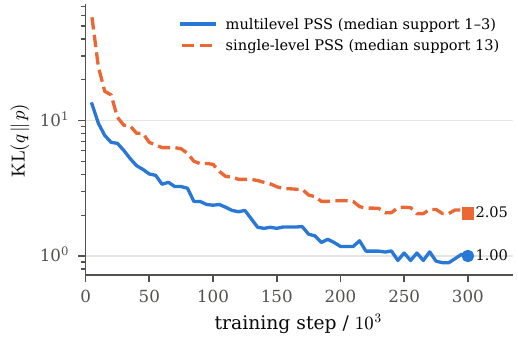}
\caption{Four-dimensional $U(1)$ at $\beta=2$: single-level against multilevel plaquette space
at matched parameter count ($\paramsSS$M against $\paramsML$M), during training. Markers are the
final evaluations of $\Neval$ samples, $\KL=\KLsingle$ and $\KLmatched$. The two models also
differ in width, depth, learning rate and seed, and the single-level model has the wider and
deeper conditioner (\cref{app:arch}).}
\label{fig:ablation}
\end{figure}

\section{Solve support and conditioning}
\label{app:support}

The layer is exact for every admissible partition $(F,D)$, but not every partition is equally
learnable. For $j\in D$ the \emph{solve support}
\begin{equation}
  k_j=\bigl|\{\,i\in F : \G_{ji}\neq0\,\}\bigr|
\end{equation}
is the row sparsity of $\G$ in \eqref{eq:solve}, $k_j=0$ meaning a plaquette copied from the coarse
field. Errors of scale $\epsilon$ in the generated angles reach a determined angle as
$O(\sqrt{k_j}\,\epsilon)$ if independent and $O(k_j\epsilon)$ if correlated, and the divergence
collects a contribution from every plaquette, determined ones included. The count of \cref{app:dof}
fixes $|F|$ and $|D|$ but is blind to $k$: two eliminations of one system can share $|F|$ and differ
entirely in support. The second quantity is $\kappa(\bm I+\G^{\!\top}\G)$, the matrix of
\eqref{eq:density} and \eqref{eq:gauss}. \Cref{tab:conditioning} is extracted from the two tables
here.

The anisotropic chain and the isotropic step generate the same $48+96+192+384=720$ angles, so the
comparison isolates locality from capacity. The isotropic and single-level solves are long-ranged,
median support $\suppiso$ and maximum $\suppisomax$; the anisotropic chain is local and becomes more
so along the chain, median $3,2,1,1$. Determinism is not the difference: the last anisotropic level
determines $75\%$ of the plaquettes, the isotropic step only $53\%$.

At matched parameter count the multilevel chain reaches half the divergence of the single-level
model, $\KLmatched$ against $\KLsingle$, with $\ESS$ $\ESSmatched$ against $\ESSsingle$
(\cref{fig:ablation}); $4^4$ is the smallest lattice on which the comparison can be made.

\begin{table}[!t]
\centering
\caption{Single-level solve of the whole lattice against one refinement level of the
coordinates of \eqref{eq:classes}, for $U(1)$. The refinement value was measured on extents $4$
and $8$ and is the same for both.}
\label{tab:dimscan}
\small
\setlength{\tabcolsep}{5pt}
\begin{tabular}{llccccc}
\toprule
 & \multicolumn{5}{c}{single level} & refinement level \\
\cmidrule(lr){2-6}\cmidrule(lr){7-7}
$d$ & lattice & generated & median support & max support & $\kappa(\bm I+\G^{\!\top}\G)$ & $\kappa$ \\
\midrule
$2$ & $4^2,\ 8^2,\ 16^2$ & $15,\ 63,\ 255$ & $15,\ 63,\ 255$ & $15,\ 63,\ 255$ & $16,\ 64,\ 256$ & $1$ \\
$3$ & $2^3,\ 4^3,\ 8^3$  & $14,\ 126,\ 1022$ & $5,\ 9,\ 17$ & $7,\ 31,\ 159$ & $20,\ 93,\ 877$ & $5$ \\
$4$ & $2^4,\ 4^4$        & $45,\ 765$ & $5,\ 9$ & $11,\ 47$ & $45,\ 379$ & $7$ \\
$5$ & ---                & --- & --- & --- & --- & $9$ \\
\bottomrule
\end{tabular}
\end{table}

In \cref{tab:dimscan} the single-level median support grows linearly with the extent and $\kappa$
roughly with the volume, so such a model must represent a Gaussian whose widths span
$\sqrt\kappa\approx20$ on $4^4$ and keep growing; on a refinement level $\kappa$ is $1,5,7,9$ in
$d=2,\dots,5$ at every extent. The single-level model of \cref{fig:ablation} uses a column priority
that gives median support $\suppiso$ and $\kappa=\kappaSingleFour$ on $4^4$ (\cref{tab:support}),
slightly above the minimum-support values of \cref{tab:dimscan}. The generated elements of
\eqref{eq:classes} are more local than the
eliminated coordinates our models use --- median support $1$, maximum $4$ and $\kappa=7$ on every
level, against $45,43,30,7$ --- and the two are related by a unimodular integer matrix
(\cref{app:verify}), so the first three refinements could be better conditioned by a change of
coordinates alone.

\subsection{Size and conditioning of the single-level solve with extent}
\label{app:solve_size}

The single-level layer solves \eqref{eq:system} for the whole lattice at once. \Cref{tab:solve_size}
gives the size of that system for the theories of this paper and for the four-dimensional
non-Abelian theories one would want next, and \cref{fig:kappa_scaling} the condition number of
\eqref{eq:gauss} against it. For $SU(N)$ the constraint Jacobian at the identity is the $U(1)$
one tensored with $\bm I_{\dim G}$, so every count is $\dim G$ times the $U(1)$ count and $\kappa$
is the same number.

\emph{Size.} The solve has $\dim G\,|D|$ rows with $|D|=1$ in $d=2$ and $|D|=3V+3$ in $d=4$:
$\solveRowsUoneThirtytwo$ for $U(1)$ and $\solveRowsSUthreeThirtytwo$ for $SU(3)$ on $32^4$. The
elimination is done once, but the layer needs its result $\G$ explicitly, $|D|\times|F|$ with a
row support that grows with the extent (\cref{tab:dimscan}, median $5,9,\dots$ at $L=2,4,\dots$),
so both the elimination and what it has to store grow faster than the lattice.

\emph{Conditioning.} In $d=2$, $\G$ is a single row of ones, $\bm I+\G^{\!\top}\G=\bm I+\bm1\bm1^{\!\top}$
has eigenvalues $V$ and $1$, and $\kappa=V$ exactly for every group and extent; the measured
values in \cref{tab:solve_size} are this. In $d=4$ the measured extents give $\kappa/V=\kappaOverVList$ at
$L=\kappaScanLList$: $\kappa$ grows linearly with the volume here too. The power law through
$L\ge4$ has exponent $\kappaAlpha$ and extrapolates to $\kappaFourSixteen$ on $16^4$ and
$\kappaFourThirtytwo$ on $32^4$. By \eqref{eq:gauss} this is the aspect ratio, squared, of the
Gaussian the flow has to represent in its own coordinates at weak coupling. The one-off exact integer elimination itself grows faster than the system: $0.1$, $4.7$, $99$ and
$824$ seconds on one CPU core at $L=4,6,8,10$ ($1{,}030$ to $40{,}006$ rows), as the fill-in grows with the
extent, so at $32^4$ ($\solveRowsUoneThirtytwo$ rows for $U(1)$) it becomes challenging with this
elimination strategy.

\emph{The refinement level.} One level of the multilevel construction has $\kappa=\kappaRefTwo$
in $d=2$ and $\kappaRefFour$ in $d=4$, support at most $4$, measured on extents $4$ and $8$ and
the same on both; the stencil is translation invariant, so neither can depend on $L$. Nothing of
extent-dependent size is eliminated or stored: the same fixed relation is applied at every site.
The count $|D|$ of determined plaquettes still grows with $V$, but as $V$ independent local
evaluations, not as one system.

% generated by ICLR_U1/figures/make_kappa_scaling.py -- do not edit by hand
\begin{table}[!htb]
\centering
\caption{Size and conditioning of the single-level solve against lattice extent, for the theories
of this paper and the four-dimensional non-Abelian theories one would want next. Counts are real
dimensions: plaquette variables $\dim G\cdot M$, rows of the solve $\dim G\cdot|D|$ with
$|D|=\operatorname{rank}\M$ in \eqref{eq:system}, generated variables $\dim G\cdot|F|$, the size of
$\bm I+\G^{\!\top}\G$. $\kappa$ is measured with minimum-support pivoting (\cref{tab:dimscan}); in
two dimensions it is $V$ exactly; $\dagger$ marks the power law of \cref{fig:kappa_scaling},
$\kappa\propto|F|^{1.05}$, extended. Last column: one refinement level in the coordinates of
\eqref{eq:classes}, the same at every extent.}
\label{tab:solve_size}
\small
\setlength{\tabcolsep}{4pt}
\begin{tabular}{lrrrrrrc}
\toprule
 & & & \multicolumn{4}{c}{single level, whole lattice} & refinement level \\
\cmidrule(lr){4-7}\cmidrule(lr){8-8}
theory & $L$ & $V$ & variables & solve rows & generated & $\kappa$ & $\kappa$ \\
\midrule
2D $U(1)$ & $8$ & $64$ & $64$ & $1$ & $63$ & $64$ & $1$ \\
 & $16$ & $256$ & $256$ & $1$ & $255$ & $256$ & $1$ \\
 & $32$ & $1{,}024$ & $1{,}024$ & $1$ & $1{,}023$ & $1{,}024$ & $1$ \\
 & $64$ & $4{,}096$ & $4{,}096$ & $1$ & $4{,}095$ & $4{,}096$ & $1$ \\
\addlinespace[2pt]
2D $SU(2)$ & $16$ & $256$ & $768$ & $3$ & $765$ & $256$ & $1$ \\
\addlinespace[2pt]
4D $U(1)$ & $4$ & $256$ & $1{,}536$ & $771$ & $765$ & $379$ & $7$ \\
 & $8$ & $4{,}096$ & $24{,}576$ & $12{,}291$ & $12{,}285$ & $7{,}088$ & $7$ \\
 & $16$ & $65{,}536$ & $393{,}216$ & $196{,}611$ & $196{,}605$ & $\approx1.3\times10^{5}{}^{\dagger}$ & $7$ \\
 & $32$ & $1{,}048{,}576$ & $6{,}291{,}456$ & $3{,}145{,}731$ & $3{,}145{,}725$ & $\approx2.4\times10^{6}{}^{\dagger}$ & $7$ \\
\addlinespace[2pt]
4D $SU(2)$ & $16$ & $65{,}536$ & $1{,}179{,}648$ & $589{,}833$ & $589{,}815$ & $\approx1.3\times10^{5}{}^{\dagger}$ & $7$ \\
 & $32$ & $1{,}048{,}576$ & $18{,}874{,}368$ & $9{,}437{,}193$ & $9{,}437{,}175$ & $\approx2.4\times10^{6}{}^{\dagger}$ & $7$ \\
\addlinespace[2pt]
4D $SU(3)$ & $16$ & $65{,}536$ & $3{,}145{,}728$ & $1{,}572{,}888$ & $1{,}572{,}840$ & $\approx1.3\times10^{5}{}^{\dagger}$ & $7$ \\
 & $32$ & $1{,}048{,}576$ & $50{,}331{,}648$ & $25{,}165{,}848$ & $25{,}165{,}800$ & $\approx2.4\times10^{6}{}^{\dagger}$ & $7$ \\
\bottomrule
\end{tabular}
\end{table}

\begin{figure}[!htb]
\centering
\includegraphics[width=3.6in]{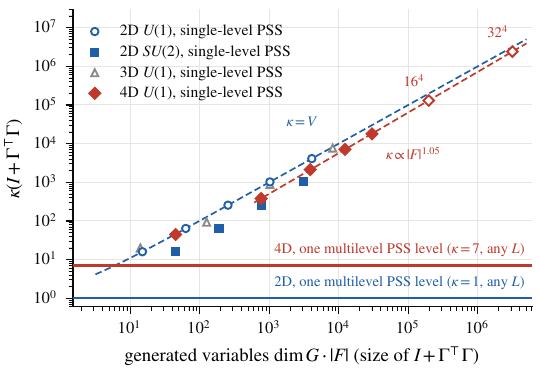}
\caption{Condition number of the single-level solve against its size, for the three theories of
this paper (3D $U(1)$ from \cref{tab:dimscan} as a reference), and one refinement level of the
multilevel construction, which is flat. Filled markers are measured; the open diamonds are the
power law through the four-dimensional points extended to $16^4$ and $32^4$; the two-dimensional
line is $\kappa=V$, exact. $SU(2)$ sits at three times the size of $U(1)$ with the same $\kappa$.}
\label{fig:kappa_scaling}
\end{figure}

\paragraph{What the determined half requires.}
$\ndet$ of the $\nplq$ plaquettes are a deterministic function of the other $\dimBtwo$, each a
signed sum of a median of $\suppiso$ of them, so a generator can reproduce every marginal it emits
and still be wrong about the half it computes: that half is fixed by high-order correlations, not by
marginals. A null with the correct marginals and no correlation, the $\dimBtwo$ angles drawn
independently from the exact one-plaquette distribution and passed through the same solve, produces
$\monoNullHalf$ wrapping defects per configuration at $\beta=0.5$ and $\monoNullTwo$ at $\beta=2$.
At $\beta=0.5$ the single-level model \emph{is} this null: $\monoSingleHalf$ defects against a
reference of $\refMonoHalf$, $\KL=\KLsingleHalf$, evidence recovery $\covSingleHalf$. At $\beta=2$
the same construction suppresses them to $\monoSingleTwo$, a factor $\monoSuppTwo$ below the null,
with $\KL=\KLsingle$ and recovery $\covSingleTwo$. The correlation the solve requires is therefore
learned in full at one coupling and not at all at the other, with no change of construction,
constraint or objective, while the multilevel sampler learns it at both, reaching $\monoHalf$
against $\refMonoHalf$ at $\beta=0.5$. The defects are magnetic monopoles, the coboundary of the
integer lifts (\cref{app:monopoles}).

\section{Verification of the construction}
\label{app:verify}

The construction makes structural claims --- that a refinement is a bijection, that its Jacobian
is one, that the solve supports are bounded --- which are exact statements and can therefore be
checked to machine precision rather than argued. \Cref{tab:verify} is that check. Every row was run
in a reference implementation written for an arbitrary compact group and dimension, for $U(1)$ with
angles over the reals, so the Abelian identities are tested exactly rather than modulo $2\pi$, and
for $SU(2)$ with unit quaternions, in double precision. No row involves training.

\begin{table}[tbh]
\centering
\caption{Numerical gates. Rows hold for both groups unless stated. ``Exact'': integer valued and
hit exactly. Other entries are the largest deviation over the lattices listed. The last row is
limited by that sampler's chart interpolant, not by the construction.}
\label{tab:verify}
\small
\setlength{\tabcolsep}{3pt}
\begin{tabular}{llr}
\toprule
what is checked & lattices & worst \\
\midrule
\multicolumn{3}{l}{\emph{Structure of one refinement}}\\
count $(d-1)(V-1)+d$; fixed links a spanning tree & $d=2$--$5$, to $4^4$  & exact \\
refinement is a bijection                     & $7$ lattices, $d=2$--$5$ & $2.3{\cdot}10^{-14}$ \\
generated element equals its plaquette        & $7$ lattices, $d=2$--$5$ & $9{\cdot}10^{-16}$ \\
cube identities hold on generated fields      & $7$ lattices, $d\ge3$     & $5.3{\cdot}10^{-15}$ \\
Haar Jacobian of the full map is one          & $d=2$--$5$, $17$--$387$ vars & $2.2{\cdot}10^{-14}$ \\
solve supports are $1,1,0,4$                  & $d=2$--$5$, extents $4,8$ & exact \\
condition number $1,5,7,9$ for $d=2$--$5$ & $d=2$--$5$, extents $4,8$ & extent independent \\
\midrule
\multicolumn{3}{l}{\emph{The measure, by importance sampling}}\\
$\log Z$ and plaquette against the closed form & $4^2,8^2,16^2$, two $\beta$ & pull $\le1.6$ \\
agreement with sampling from Haar links        & $2^3,\,4{\times}2^2,\,2^4$    & pull $\le1.7$ \\
$\log Z$ against thermodynamic integration     & $2^4$, $\beta=0.2,\,0.5$      & pull $\le0.7$ \\
\midrule
\multicolumn{3}{l}{\emph{The implementations used in the experiments}}\\
each level is a unimodular change of variables & five levels of $2^4\to4^4$ & exact \\
solves give back the coarse field and cube identities & five levels & $7{\cdot}10^{-15}$ \\
two-dimensional block equals two refinements   & $L_{\mathrm c}=2,4,8$      & exact \\
$SU(2)$ block relation, inverse and closure    & $L\le6$                    & $5{\cdot}10^{-14}$ \\
$SU(2)$ reconstructed links are Haar           & $L\le6$                    & $2.1{\cdot}10^{-3}$ \\
$SU(2)$ block log-density against its Jacobian & $L\le6$                    & $7{\cdot}10^{-7}$ \\
\bottomrule
\end{tabular}
\end{table}

\paragraph{How the Jacobian is measured.}
Each generated element $F$ is perturbed as $e^{\omega}F$ with $\omega$ in the Lie algebra, and each
non-tree link $U$ is read as $e^{\omega'}U$. Haar measure is bi-invariant, so its density in these
coordinates is the same at every point, and the Haar Jacobian of the map is the determinant of
$\partial\omega'/\partial\omega$ at $\omega=0$, computed by automatic differentiation at two random
points per lattice.

\paragraph{How the measure is tested.}
The importance-sampling rows generate every level from a conditional proposal and reweight by
$e^{-S}/q$ with \eqref{eq:chain}, so they test the density bookkeeping and the measure, not only
the geometry. The proposal draws each $F_\nu(X)$ from the law of its two electric terms in
\eqref{eq:classes}, von Mises for $U(1)$ and von Mises--Fisher on $S^3$ for $SU(2)$. On a
two-dimensional torus the Haar-normalised partition function is known in closed form,
\begin{equation}
  Z=\sum_{n\in\Z}I_n(\beta)^V \quad [U(1)],
  \qquad
  Z=\sum_{n\ge1}\big(2I_n(\beta)/\beta\big)^V \quad [SU(2)],
\end{equation}
with $S=-\beta\sum\tfrac1N\operatorname{Re}\operatorname{tr}P$: at $L=16$ the estimate is
$210.9423\pm0.0014$ against $210.9424$ for $U(1)$ at $\beta=2$, and $248.0914\pm0.0009$ against
$248.0905$ for $SU(2)$ at $\beta=3$. Where no closed form exists the same estimator is compared
with importance sampling from Haar-distributed links, which uses no gauge fixing and none of the
construction, and on $2^4$ for $U(1)$ with the thermodynamic-integration $\log Z$ of
\cref{app:reference}, which additionally tests the unimodularity of the coarse level. A row counts
only if its estimate has at least $1{,}000$ effective samples.

\section{Architecture}
\label{app:arch}

\paragraph{The model is a stack of conditional flows.}
Level $i$ draws its free variables from a learned conditional density and the constraint layer
turns them into a valid configuration, $F^{(i)}\sim\widetilde q_i(\cdot\mid P^{(i-1)})$,
$P^{(i)}=R_i(F^{(i)};P^{(i-1)})$. Because $R_i$ is a bijection of unit Haar Jacobian
\eqref{eq:coords}, the density transfers without a correction,
\begin{equation}
  \log q_i\big(P^{(i)}\mid P^{(i-1)}\big)=\log\widetilde q_i\big(F^{(i)}\mid P^{(i-1)}\big),
  \qquad
  \log q(P)=\sum_{i=0}^{n}\log\widetilde q_i\big(F^{(i)}\mid P^{(i-1)}\big),
  \label{eq:logq}
\end{equation}
so the exact log-density needed by the objective \eqref{eq:rkl} and by the importance weights costs
no more than the per-level flow densities. Each $\widetilde q_i$ is a normalizing flow on
$\T^{|F^{(i)}|}$ conditioned on the coarse configuration; the coarse level $i=0$ is unconditional.
Learning is therefore a conditional density-estimation problem per resolution, and the physics
enters only through $R_i$, which has no parameters.

\paragraph{Coupling layers.}
$\widetilde q_i$ is a uniform base composed with a stack of coupling layers (their number is in \cref{tab:arch}). Layer $\ell$ partitions the
components into active $A_\ell$ and passive $\bar A_\ell$ and maps
\begin{equation}
  f_j\mapsto \mathrm{RQS}_K\big(f_j;\theta_j\big)\ \ (j\in A_\ell),
  \qquad
  \theta=c_\phi\big(\cos f_{\bar A_\ell},\,\sin f_{\bar A_\ell},\,P^{\mathrm c}\big)\in\R^{3K},
  \label{eq:coupling}
\end{equation}
with $\mathrm{RQS}_K$ the circular rational-quadratic spline \citep{Durkan2019} on $K$ bins,
matched in value and derivative at $\pm\pi$, and $c_\phi$ the conditioner. Its last layer is
zero-initialised, so each coupling starts at the identity and the stack starts at the base. For
$SU(2)$, \eqref{eq:coupling} acts on the angles of an equivariant parametrisation over a Haar base.

\paragraph{Why the split matters.}
A coupling layer represents the dependence of $A_\ell$ on $\bar A_\ell$ but none among $A_\ell$, so
the partitions decide what the stack can express, and two requirements act on them. Every component
must be active in some layers and passive in others, which is what the alternation in \cref{tab:arch}
provides. And the constraint couples the free variables: a determined plaquette is a function of
$k_j$ of them, \eqref{eq:solve}. If two of those inputs are active in the same layer, that layer
cannot represent their joint effect on it, which is why the constrained splits admit at most one
active input per determined plaquette. This is the only point at which the architecture is aware of
the constraint at all.

\paragraph{Conditioner.}
$c_\phi$ is dense on the $2^4$ coarse level and in $d=2$, where it sees the coarse plaquette of the
block and its neighbours, and for $SU(2)$ the holonomy commutator as well. On a $d=4$ refinement it
is a convolution on the fine lattice: free angles scattered to their plaquette positions, $36$ input
channels --- $\cos f$, $\sin f$ on $\bar A_\ell$, active and passive indicators, and
$\cos P^{\mathrm c}$, $\sin P^{\mathrm c}$ repeated along the refinement direction, six plaquette
orientations each --- a star stencil of radius $1$ applied by periodic shifts, and SiLU
activations. The shifts make $c_\phi$ exactly equivariant under lattice translations, so the
symmetry of the target is built in rather than learned, and the parameter count is independent of
the volume.

\paragraph{What the chain buys the model.}
The single-level model is the same machinery at $n=0$: one unconditional flow over all $|F|$ free
variables, followed by the global solve. The difference the chain makes is locality of the learning
problem. On a refinement every determined plaquette is a function of at most four generated
variables, with condition number $7$, in the coordinates of \cref{app:refine}; in the eliminated
coordinates our models use it depends on at most nine, with $\kappa\le45$ (\cref{tab:support}),
and every such group of inputs lies within lattice distance three of one another. Both are
independent of the extent, so an equivariant convolution of three or four radius-one layers,
conditioned on the coarse field, sees every dependence the constraint creates, and one network
serves every level and every volume. Solved on the whole
lattice at once the median support is $\suppiso$ with maximum $\suppisomax$ and grows with the
extent, and $\kappa=\kappaSingleFour$ on $4^4$: no fixed stencil reaches the variables a determined
plaquette couples, and the flow has to build those correlations implicitly through depth. The
coarse configuration enters as conditioning input and carries the long-range structure, so each
level has only to learn the degrees of freedom its own refinement adds. This is the architectural
content of the construction, and \cref{fig:ablation} tests it at matched parameter count.

\begin{table}[tbh]
\centering
\caption{Our models. Width/depth is that of $c_\phi$ in \eqref{eq:coupling}. Couplings run coarse
level first, then each refinement; one entry means the same count at every level. Splits:
\emph{halves} of a random permutation; \emph{parity} over sites times a per-layer orientation
pattern; \emph{constrained}, no determined plaquette with two active inputs in a coupling, whose
models also pass $c_\phi$ each partially determined plaquette's running sum and distance to
$\pm\pi$. At $\beta=0.5,0.8,2.0,3.0$ four further stacks of eight couplings act on the finished
field, one per direction and conditioned on it blocked along that direction, an exact change of
variables.}
\label{tab:arch}
\small
\setlength{\tabcolsep}{3pt}
\begin{tabular}{lllcccll}
\toprule
theory & levels & model & $\beta$ & $K$ & width/depth & couplings & split \\
\midrule
4D $U(1)$ & $2^4\!\to\!4^4$            & multilevel & $0.5,\,0.8$ & 8 & $64$/$3$  & $24$; $18,18,16,12$ & constrained \\
          & $2^4\!\to\!4^4$            & multilevel & $1.0$       & 8 & $64$/$3$  & $8$                 & halves \\
          & $2^4\!\to\!4^4$            & multilevel & $1.5$       & 8 & $128$/$4$ & $8$                 & halves \\
          & $2^4\!\to\!4^4$            & multilevel & $2.0,\,3.0$ & 8 & $128$/$4$ & $8$                 & parity \\
          & $4^4$, single              & single     & all         & 8 & $128$/$4$ & $8$                 & halves \\
\cmidrule(lr){1-8}
2D $U(1)$ & $2^2\!\to\!32^2$, $L{=}4$  & multilevel & $3.0$       & 8 & $64$/$3$  & $6$; $8$ per level  & halves \\
          & $2^2\!\to\!64^2$, $L\ge8$  & multilevel & $0.1875L^2$ & 8 & $64$/$3$  & $6$; $8$ per level  & halves \\
          & $L^2$, single, $L\le8$     & single     & $0.1875L^2$ & 8 & $8$/$2$   & $48$                & halves \\
\cmidrule(lr){1-8}
2D $SU(2)$ & $2^2\!\to\!16^2$            & multilevel & all         & 8 & $64$/$3$  & $32$; $8$ base      & halves \\
           & $16^2$, single             & single     & all         & 8 & $64$/$3$  & $32$                & halves \\
\bottomrule
\end{tabular}
\end{table}

\begin{table}[tbh]
\centering
\caption{Baselines. Link-space flows are trained against the same action as our models. The
single-level plaquette flow is our own construction and is in \cref{tab:arch}.}
\label{tab:arch_base}
\small
\setlength{\tabcolsep}{5pt}
\begin{tabular}{lll}
\toprule
theory & baseline & construction \\
\midrule
2D $U(1)$  & link multiscale \citep{Abbott2023}          & coarse $L/2$ prior, then a fine flow \\
2D $U(1)$  & link single level \citep{Kanwar_2020}        & one flow on all links \\
4D $U(1)$  & link multiscale \citep{Abbott2023}          & as published, values from Fig.~4 \\
2D $SU(2)$ & link continuous flow \citep{Gerdes:2024rjk} & ODE on links, \texttt{bijx}, $40$ steps \\
\bottomrule
\end{tabular}
\end{table}

The $2$D rows use the refinement of \cref{app:block2d}, with $c_\phi$ seeing the coarse plaquette
of the block and its neighbours, and for $SU(2)$ the holonomy commutator. The $SU(2)$ levels apply
their $32$ flow steps over $16$ group-element blocks with an interleave of $4$, the base level $8$
steps over $4$; holonomies are drawn from the Haar base, not learned. The $4$D coarse level is
dense on its $45$ free angles; the single-level model generates all $\dimBtwo$ free angles of $4^4$
through the solve of \cref{app:support} and is the comparison in \cref{fig:ablation}.

\section{Training protocol}
\label{app:training}

All models, ours and the link-space baselines we trained, minimise the reverse Kullback--Leibler
divergence \eqref{eq:rkl} with pathwise gradients through the flow and the constraint layer. No
Monte Carlo configuration enters training; the ensembles of \cref{app:reference} are used only for
validation. Adam, double precision. The $4$D chain and the $SU(2)$ chain are trained as a whole, every level jointly from random
initialisation. The $2$D $U(1)$ chain is trained as a ladder along $\beta=cL^2$: one level at a
time, each warm-started with the levels below it frozen, so the budget is per level. Schedules are
set by the step target, and runs interrupted by cluster limits resume optimiser and scheduler state
exactly.

\begin{table}[t]
\centering
\caption{Training settings for every model reported in \cref{sec:exp}. ``exp'' decays the learning
rate exponentially from the first value to the second over the scheduled steps. Two-dimensional
budgets for the $U(1)$ ladder are per level, in the order root, $L=4,8,16,32,64$. Gradients are clipped
at $1$, except the $4$D multilevel model at $\beta=2$ at $10$ and the baselines at $0.5$. Each $SU(2)$ single-level
run has converged or diverged well inside $25{,}000$ steps, so the budget is not what limits it.}
\label{tab:training}
\small
\setlength{\tabcolsep}{2pt}
\begin{tabular}{llrrll}
\toprule
theory & model & steps & batch & learning rate & sched. \\
\midrule
4D $U(1)$, $\beta{=}0.5,0.8$ & multilevel   & $260{,}000$ & $192$ & $3{\cdot}10^{-4}\!\to\!10^{-5}$    & exp \\
4D $U(1)$, $\beta{=}1.0,1.5$ & multilevel   & $300{,}000$ & $192$ & $3{\cdot}10^{-4}\!\to\!5{\cdot}10^{-6}$ & exp \\
4D $U(1)$, $\beta{=}2.0$     & multilevel   & $300{,}000$ & $192$ & $3{\cdot}10^{-4}\!\to\!5{\cdot}10^{-6}$ & exp \\
4D $U(1)$, $\beta{=}3.0$     & multilevel   & $300{,}000$ & $192$ & $3{\cdot}10^{-4}\!\to\!5{\cdot}10^{-6}$ & exp \\
4D $U(1)$, all $\beta$       & single level & $300{,}000$ & $192$ & $3{\cdot}10^{-4}\!\to\!5{\cdot}10^{-6}$ & exp \\
2D $U(1)$, $L{=}4$--$64$     & multilevel   & $20/50/30/30/50/45$k & $256$ & $10^{-3}$ & ladder \\
2D $SU(2)$, all $\beta$      & multilevel   & $100{,}000$        & $128$ & $3{\cdot}10^{-4}$ & const. \\
2D $SU(2)$, all $\beta$      & single level & $25{,}000$, converged & $128$ & $3{\cdot}10^{-4}$ & const. \\
\cmidrule(lr){1-6}
2D $U(1)$ & link multiscale, baseline   & $80{,}000$ & $512$ & $10^{-3}$ & const. \\
2D $U(1)$ & link single level, baseline & $40{,}000$ & $512$ & $10^{-3}$ & const. \\
2D $U(1)$ & plaquette single level      & $80{,}000$ & $512$ & $10^{-3}$ & const. \\
\bottomrule
\end{tabular}
\end{table}

Checkpoints are written during training together with a training-time estimate $\widehat{\KL}$ on
fresh samples, every $5{,}000$ steps on $65{,}536$ samples in $d=4$, or every $10{,}000$ on
$16{,}384$ at $\beta=0.5,\,0.8,\,3.0$; the same schedule is used for the baselines. One retained
checkpoint per model is evaluated, and every number in the paper is recomputed from one fresh
evaluation of $\Neval$ samples (\cref{app:estimation}), never from the training-time estimate.

The settings of \cref{tab:training} are target dependent: width, couplings, clipping and schedule
were chosen per coupling, and we claim no single recipe across the transition and no wall-clock
advantage over Markov chain Monte Carlo (\cref{app:estimation}).

\section{Reference values and the partition function}
\label{app:reference}

No Monte Carlo configuration enters training. The quantities here validate the trained samplers:
reference expectation values for the reweighted observables, and the $\log Z$ that turns the
estimated divergence into the true one.

\paragraph{Protocol.}
Reference ensembles are generated in link variables by heatbath, each link drawn exactly from its
conditional von Mises distribution given its six staples, the two site-parity sublattices of a
direction updated in turn so that every half-step is an exact block update. At each coupling eight
chains of $30{,}000$ sweeps are run, four from the ordered configuration and four from uniformly
random links; the first $6{,}000$ are discarded and every later sweep is measured. The quoted error
is the larger of the autocorrelation-corrected error of the pooled mean and the standard error of
the chain means, so neither an underestimated autocorrelation nor a confined chain can flatter it.
\begin{center}
\small
\begin{tabular}{ccccc}
\toprule
$\beta$ & $\tau_{\mathrm{int}}$ (sweeps) & $\langle\cos\varphi\rangle$ & monopoles per configuration & $\log Z$ \\
\midrule
$0.5$ & $0.6$ & $\refPlqHalf$          & $\refMonoHalf$      & $732.941 \pm 0.009$ \\
$0.8$ & $1.1$ & $0.402005 \pm 0.000066$ & $250.18 \pm 0.06$   & $420.204 \pm 0.011$ \\
$1.0$ & $15$  & $\refPlqOne$           & $\refMonoOne$       & $262.052 \pm 0.019$ \\
$1.5$ & $0.7$ & $\refPlqOneHalf$       & $\refMonoOneHalf$   & $75.942 \pm 0.021$ \\
$2.0$ & $0.6$ & $\refPlqTwo$           & $\refMonoTwo$       & $-45.114 \pm 0.021$ \\
$3.0$ & $0.6$ & $0.912952 \pm 0.000011$ & none observed       & $-210.003 \pm 0.021$ \\
\bottomrule
\end{tabular}
\end{center}
Only $\beta=1$, next to the transition, has a long autocorrelation. There ordered and disordered
chains agree after $\hbMerge$ sweeps, two orders of magnitude inside the $6{,}000$ discarded, and
the plaquette distribution is a single broad peak of width $\cosSdOne$ with monopole spread
$\monoSdOne$: no sign of two coexisting phases on this volume.

\paragraph{Flux sectors.}
Heatbath chains never change the flux through the six two-tori. A unit flux costs an action of
about $\fluxCost$, suppressed by $e^{-30}$ at $\beta=1.5$ and $e^{-39}$ at $\beta=2$, so a chain
starting in one is a metastable artefact rather than an equilibrium sample, and is visibly
displaced: $\langle\cos\varphi\rangle=0.7998$ against $0.8134$ at $\beta=1.5$. The reference at
$\beta\ge1.5$ therefore uses trivial-sector chains only --- five of eight at $\beta=1.5$, seven of
eight at $\beta=2$, all eight at $\beta=3$ --- one of the three dropped at $\beta=1.5$ going
instead for a flux made ambiguous by a monopole, its mean agreeing with the rest. At $\beta\le1$
monopoles are dense and flux is not well defined, so all chains are used. The flow can generate
every sector (\cref{app:constraint}).

\paragraph{Cross-check in the sampler's own coordinates.}
To exclude an error in the target as the flow sees it, Metropolis updates were run on
$p(f)\propto e^{-S(\varphi(f))}$ over the $\dimBtwo$ free angles of the single-level $4^4$
parametrisation at $\beta=1$: sixteen chains of $3{,}000$ sweeps from $f=0$ and from uniform $f$
give $\langle\cos\varphi\rangle=\mcColdPlq$ and $\mcHotPlq$ with $\mcColdMono$ and $\mcHotMono$
monopoles per configuration, each to about $0.008$, both agreeing with the heatbath reference. The
target and its coordinates are therefore correct, including the monopole sector.

\paragraph{Exact \texorpdfstring{$\log Z$}{log Z}.}
Throughout this appendix and in the training code the action is written
$S=\beta\sum_p(1-\cos\varphi_p)$, which is \eqref{eq:action} plus the constant $\beta\binom d2V$;
the tabulated $\log Z$ are for this form and for the angle measure $\mathrm df$ on the $\dimBtwo$ free
angles, whereas \eqref{eq:action} uses normalised Haar measure, $\mathrm df/2\pi$ per angle; they
therefore differ from those of \eqref{eq:action} by $-\nplq\,\beta+\dimBtwo\log2\pi$, while every
divergence and weight is unchanged. The flow's importance weights
estimate
$Z(\beta)=\int_{T^{\dimBtwo}}e^{-S(\varphi(f))}\,\mathrm{d}f$, whose logarithm follows from
thermodynamic integration,
\begin{equation}
  \log Z(0)=\dimBtwo\,\log2\pi,
  \qquad
  \frac{\mathrm{d}\log Z}{\mathrm{d}\beta}=-\nplq\,\big(1-\langle\cos\varphi\rangle_\beta\big),
\end{equation}
integrated by the trapezoid rule on a grid of spacing $0.02$, refined to $0.005$ across the
transition. Each grid point uses sixty-four ordered-start chains contributing $1{,}600$
measurements after thermalisation; the error, from the spread of chain means propagated with the
trapezoid weights, is at most $0.021$ in $\log Z$, hence at most that in the true divergence and
$2\%$ in the evidence recovery.

\section{Estimators and evaluation protocol}
\label{app:estimation}

\paragraph{One evaluation per model.}
Each model is evaluated once, on $N=\Neval$ independent samples from its retained checkpoint with a
fixed seed unless stated otherwise (the link-space continuous flow for $SU(2)$ at $N\le10^4$,
\cref{sec:exp}; the Wilson loops of \cref{tab:su2_wilson} at $2\times10^5$), and every number reported for it --- divergence, effective sample size, top weight,
evidence recovery and every reweighted observable --- comes from that one sample. Evaluations are
never combined, training-batch statistics are never quoted, and any departure is stated where the
number is used.

\paragraph{Weights, divergences and errors.}
For samples $f_i\sim q$ on the free angles, $\log\tilde w_i=-S(\varphi(f_i))-\log q(f_i)$, the
constant Jacobian of \eqref{eq:density} cancelling throughout. Then
\begin{equation}
  \ESS=\frac{\big(\sum_i\tilde w_i\big)^2}{N\sum_i\tilde w_i^2},\qquad
  t=\frac{\max_i\tilde w_i}{\sum_i\tilde w_i},\qquad
  \log\hat Z=\log\frac1N\sum_i\tilde w_i,
\end{equation}
\begin{equation}
  \widehat{\KL}=-\frac1N\sum_i\log\tilde w_i+\log\hat Z,\qquad
  \KL=-\frac1N\sum_i\log\tilde w_i+\log Z,\qquad
  R=e^{\log\hat Z-\log Z}=e^{\widehat{\KL}-\KL},
\end{equation}
with the $\log Z$ of \cref{app:reference}. Expectations are self-normalised,
$\hat O=\sum_i\tilde w_iO_i/\sum_i\tilde w_i$; errors on $\ESS$, $\KL$ and every reweighted
observable are delete-block jackknife over the same $N$ draws in $100$ blocks, that on $\KL$
including the uncertainty of $\log Z$. The observable panels all report the relative accuracy $\hat O/O_{\mathrm{ref}}-1$ of
\eqref{eq:diagnostics}: against the exact finite-volume value in \cref{fig:scaling_te,fig:su2_kl_obs_ess}, and
against the heatbath reference in \cref{fig:u1_4d}, in units of $10^{-3}$ with the reference's own
uncertainty as a band. The monopole number uses
all $N$ samples, since at $\beta=2$ only about one configuration in a hundred carries one.

$\widehat{\KL}$ is the divergence available without $\log Z$, biased low by exactly $-\log R$ ---
a large gap when a model misses much of the target, $\widehat{\KL}=\KLoneHatModel$ against $\KL=\KLone$ at
$\beta=1$. Every divergence in the paper is the true one. A flow has full support, so
$\mathbb{E}_q[\tilde w]=Z$ and $R\to1$ as $N\to\infty$ for every model; at finite $N$, $R$ says
how much of the evidence the drawn samples recovered, and may sit slightly above one, as at
$\beta=0.5$ ($R=\ZrecHalf$). With $\log Z$ known to $0.021$, $R$ is limited by the sample, not the
reference.

\paragraph{Where the normalisation comes from.}
$\KL(q\|p)=\mathbb E_q[\log q+S]+\log Z$: the first term is a sample mean, the second a constant
whose source decides what the reported number means. We take it from thermodynamic integration, so
it never involves the model. The link-space baseline we compare against estimates it from the model
instead, as $\log\widehat Z=\log\overline{\tilde w}$ \citep[Fig.~4]{Abbott2023}. The two differ by
exactly the log of the evidence recovery,
\begin{equation}
  \KL_{\widehat Z}=\KL-\log(1/R),
  \label{eq:logz_convention}
\end{equation}
and by Jensen's inequality the model-based convention is a lower bound in expectation, short by
$\log(1/R)$ nats, so it flatters a model precisely where the model covers the target least well.
\Cref{fig:logzconv} shows the effect on our own sampler: away from the transition the two agree to
a few per cent, while at $\beta=1$, where $R=\ZrecOne$, the reported divergence falls from $\KLone$
to $\KLoneHatModel$, suppressing the peak by a factor of three. The published link-space values
cannot be corrected, their $\log\widehat Z$ not being reported; the direction of the effect makes
our comparison against them conservative.

\begin{figure}[t]
\centering
\begin{minipage}[c]{0.56\linewidth}
  \includegraphics[width=\linewidth]{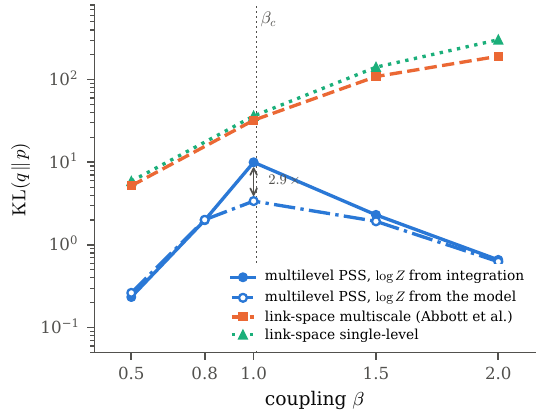}
\end{minipage}\hfill
\begin{minipage}[c]{0.40\linewidth}
  \caption{The same four-dimensional models and evaluations with the normalisation taken two ways:
  filled markers use $\log Z$ from thermodynamic integration, open markers the model's own
  $\log\widehat Z$. The vertical gap is $\log(1/R)$ and nothing else, so it closes where the sampler
  covers the target and opens where it does not. A missed mode contributes nothing to the
  model-based quantity, which is how a transition peak can be flattened by the convention alone.}
  \label{fig:logzconv}
\end{minipage}
\end{figure}

\paragraph{What the effective sample size shows, and the \texorpdfstring{$n_{\mathrm{eff}}$}{n\_eff} cut.}
Since $\sum_i\tilde w_i^2\ge\max_i\tilde w_i^2$, the heaviest single sample caps it,
$N\cdot\ESS\le1/t^2$: at $\beta=1$, $t=\topOne$ caps the effective number of samples near $540$ and
the measured value is about $170$. When $t$ is far above $1/N$ the effective sample size rests on a
handful of samples and is itself noisy, so $t$ is reported next to it and an estimate is treated as
a measurement only when no sample dominates. More importantly, it measures how evenly the weight is
spread over the region the model reaches, not whether the model reaches the whole target --- a model
confined to part of it can have nearly uniform weights there --- so every result is reported with
its true divergence, evidence recovery and comparison against the reference, and the effective
sample size is never used on its own as evidence of correctness. Every point is drawn; where
$n_{\mathrm{eff}}=N\cdot\ESS<100$ the weighted estimators rest on a few samples, since a single
dominant weight invalidates them, whereas the divergence $\KL(q\|p)=\log Z-\overline{\log w}$ is
an unweighted mean over draws from $q$, unbiased at any $n_{\mathrm{eff}}$, with uncertainty
$\Var(\log w)/N$ that the sample sizes used here resolve.

\section{Monopoles}
\label{app:monopoles}

For $U(1)$ the solve \eqref{eq:solve} is exact over $\R$, while the physical plaquette is
$\wrap\varphi\in(-\pi,\pi]$. With the integer lifts $k:=(\varphi-\wrap\varphi)/2\pi$ and
$\mathrm d\varphi=0$,
\begin{equation}
  m:=\frac{1}{2\pi}\,\mathrm{d}\big(\wrap\varphi\big)=-\,\mathrm{d}k ,
  \qquad
  Q:=\frac1{2\pi}\sum_x\wrap\varphi_{01}(x)=-\sum_x k_{01}(x)\quad(d=2),
  \label{eq:monopole}
\end{equation}
so the monopole charge of a three-cube is the coboundary of the integer lifts of its faces, and in
two dimensions the topological charge is their sum. Generated angles lie in $(-\pi,\pi]$ and have
$k=0$, so a monopole is produced exactly when the real-valued solve of a determined plaquette
leaves $(-\pi,\pi]$, crossing the seam at $\pm\pi$. The sector is therefore inside the parametrised
set rather than excluded by it: uniformly random free angles give $\monoRandom$ monopoles per
configuration on $4^4$, more than the reference at any coupling studied.
How closely the trained flow reproduces their density at weak coupling is a question of sampling rare
defect configurations, which is beyond the scope of this work (\cref{app:limitations}).

% COMMENTED OUT, file kept. Nothing the paper reports depends on it: no winding observable
% (Polyakov loop or correlator) appears in the main text or in any result, so this appendix
% described machinery no number uses. Its general content -- that a maximal tree recovers links
% from plaquettes -- is now one sentence at the end of app:dof.tree. Restoring it means restoring
% this \input AND the \cref{app:links} citations in app:dof.tree and app:su2single.
%\input{appendix/12_link_reconstruction}
\section{Two dimensions}
\label{app:twod}

The same constraint layer, with the block relation of \cref{app:block2d} as the linear system, gives
a multilevel sampler for two-dimensional compact $U(1)$: the coupling layers of \cref{app:arch} and
the objective of \cref{app:training}, an unconditional base on a $2\times2$ lattice, then
power-of-two doublings to $L$, each new level warm-started with the lower levels frozen. Every $L$
is trained along a line of constant physics, $\beta=cL^2$, so a ladder describes one continuum
theory at successive resolutions.

\paragraph{Topology is not what this construction addresses.}
The difficulty of the flow is set by how active the topology is, the spread of $Q$, rather than by
$\beta$ directly: at small $c$, where $\langle Q^2\rangle$ is large, the
target is strongly multimodal in $Q$ and the mode-seeking reverse-KL flow loses $\ESS$, while where
the topology is mild ($\langle Q^2\rangle<1$ at every level) $\ESS$ is high and flat across scales.
This is a coverage problem in the topological variable, not a defect of the constraint layer, whose
exactness is independent of $c$. We do not address it here; methods that resolve the sectors explicitly, such
as the sector-resolved flow sampling of \citet{Singha:2026aac}, are complementary to the
construction of this paper and could be combined with it.

\paragraph{Topological susceptibility.}
The plaquette is strongly self-averaging and therefore a weak test of a sampler: even weights
dominated by a few configurations reproduce it to parts in $10^5$. The susceptibility is the
sensitive observable, and on the torus it is known exactly from the character expansion at every
$(\beta,V)$, so no reference simulation is needed. \Cref{fig:chitop} shows the reweighted
$\chi_{\mathrm{top}}$ against that value along $\beta=\teC\,L^2$. The multilevel sampler agrees
within $0.7\sigma$ at every size, the largest deviation being $+0.5\%$ at $L=32$. Points whose weights are degenerate are omitted, their errors being too wide to test; the
single-level plaquette flow at $L=8$ is drawn at $-100\%$, its $\chi_{\mathrm{top}}$ identically
zero with every configuration at $Q=0$, which is the clearest statement that imposing the closure
constraint does not by itself generate topology. That this is topological rather than a property of
the global solve is settled by $SU(2)$: two dimensions give it the same single closure row and hence
the same $\kappa=V$, since $\G_{SU(N)}=\G_{U(1)}\otimes\bm I_{\dim G}$, but it is simply connected
and has no sectors to collapse into, and its single-level sampler is efficient at $L=16$
(\cref{fig:su2_kl_obs_ess}), degrading only at long correlation length and through a weight tail
rather than a bias.

\begin{figure}[tbh]
\centering
\begin{minipage}[c]{0.50\linewidth}
  \includegraphics[width=\linewidth]{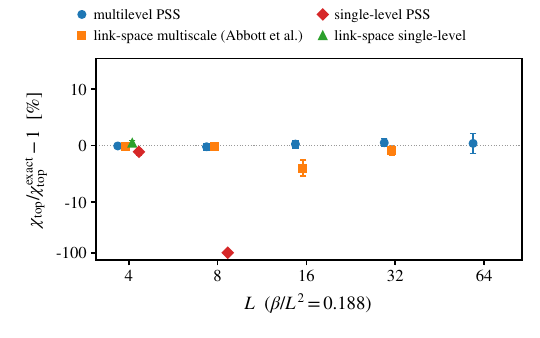}
\end{minipage}\hfill
\begin{minipage}[c]{0.46\linewidth}
  \caption{Reweighted topological susceptibility against the exact character-expansion value along
  $\beta=\teC\,L^2$, as a relative deviation. Samplers whose weights are degenerate are omitted; the
  single-level plaquette flow at $L=8$ has $\chi_{\mathrm{top}}\equiv0$ and is drawn at $-100\%$.}
  \label{fig:chitop}
\end{minipage}
\end{figure}
% GENERATED by ICLR_U1/su2/figures/make_appendix_su2.py -- do not edit by hand.
% Source records: ICLR_U1/su2/REPRO/appendix_wilson/ (18 result JSONs + generator).
\begin{table}[tbh]
\centering
\footnotesize
\caption{Wilson loops for the two plaquette-space samplers on $16^2$ $SU(2)$,
reweighted at $N=2\times10^5$ and compared with the exact values.  Entries are
$10^5(\langle W\rangle-W_\text{exact})$, with the deviation in standard deviations
in parentheses; $W_\text{exact}$ is listed because the loop decays with area.
The action is the $1\times1$ loop, so $2\times2$ and
$3\times3$ are observables neither sampler was trained to reproduce.}
\label{tab:su2_wilson}
\begin{tabular}{rr rrr rrr}
\toprule
 & & \multicolumn{3}{c}{multilevel PSS} & \multicolumn{3}{c}{single-level PSS} \\
\cmidrule(lr){3-5}\cmidrule(lr){6-8}
$\beta$ & $\xi/a$ & $1\times1$ & $2\times2$ & $3\times3$ & $1\times1$ & $2\times2$ & $3\times3$ \\
\midrule
$2.2$ & $1.1$ & $-11.8\,(2.2)$ & $-9.7\,(1.2)$ & $+0.4\,(0.1)$ & $-9.0\,(1.6)$ & $-15.6\,(1.8)$ & $+6.8\,(0.9)$ \\
$8$ & $2.2$ & $-0.6\,(0.3)$ & $-1.8\,(0.3)$ & $-3.5\,(0.3)$ & $-1.9\,(0.9)$ & $-4.5\,(0.5)$ & $+4.4\,(0.3)$ \\
$18$ & $3.4$ & $+0.2\,(0.2)$ & $+6.1\,(1.1)$ & $+8.4\,(0.7)$ & $-1.1\,(1.0)$ & $-7.6\,(1.2)$ & $-32.0\,(1.9)$ \\
$40$ & $5.1$ & $-0.6\,(0.8)$ & $-0.4\,(0.1)$ & $+2.7\,(0.3)$ & $-1.1\,(0.5)$ & $-14.8\,(1.5)$ & $-45.2\,(1.9)$ \\
$65.9135$ & $6.6$ & $-0.2\,(0.5)$ & $-0.0\,(0.0)$ & $+2.0\,(0.4)$ & $+0.7\,(0.8)$ & $+3.0\,(0.6)$ & $+4.4\,(0.4)$ \\
\midrule
\multicolumn{8}{l}{\footnotesize $W_\text{exact}$ at $1\times1$, $2\times2$, $3\times3$:} \\
\multicolumn{8}{l}{\footnotesize $\beta=2.2$: $0.4645$, $0.0465$, $0.0010$;\quad $\beta=8$: $0.8192$, $0.4505$, $0.1662$} \\
\multicolumn{8}{l}{\footnotesize $\beta=18$: $0.9179$, $0.7099$, $0.4625$;\quad $\beta=40$: $0.9627$, $0.8591$, $0.7105$} \\
\multicolumn{8}{l}{\footnotesize $\beta=65.9135$: $0.9773$, $0.9124$, $0.8136$} \\
\bottomrule
\end{tabular}
\end{table}
   % placed early so it floats onto the same page as the text

\paragraph{Protocol, and the check against the exact value.}
Every point of \cref{fig:scaling_te} is one evaluation with seed $0$, and every size is a stage
of the same ladder, with no averaging over seeds and no selection among draws. Used as an independence-Metropolis proposal the sampler has a topological autocorrelation time
near unity, $\tauQrange$ for $L\le32$. At the largest extent of the
main-text ladder, $L=64$ ($\beta=768$), the multilevel sampler reaches $\ESS=0.101\pm0.006$ after
$45{,}000$ steps at the finest level, with plaquette and susceptibility within $0.2\sigma$ of the
exact values and $\KL(q\|p)=3.9$ nats against $17$ for the link-space multiscale flow.

\paragraph{$SU(2)$ Wilson loops.}
\Cref{tab:su2_wilson} evaluates both plaquette-space samplers on Wilson loops of area $1$, $4$ and
$9$. The action is the $1\times1$ loop, so $\langle\tfrac12\mathrm{tr}\,U_p\rangle$ is essentially
the training target, whereas the $2\times2$ and $3\times3$ loops are observables neither sampler
was trained to reproduce and which two-dimensional $SU(2)$ supplies in closed form. All thirty
measurements agree with the exact values to within $2.2\sigma$, so the reweighted estimates are consistent with the exact values, within uncertainties,
beyond the quantity the flow optimises. The failure mode at weak coupling is a heavy weight tail
rather than a displaced bulk: at $\xi/a=6.6$ the largest single weight carries $1.5\%$ of the
single-level estimator and $48\%$ of the link-space one, against at most $4.4\times10^{-4}$ for the
multilevel sampler at any coupling. Effective sample sizes and reweighted observables are
draw-dependent at these volumes, so each plaquette-space number quoted in \cref{fig:su2_kl_obs_ess} is one
evaluation at $N=\Neval$ at a fixed seed (the link-space flow at $N\le10^4$, \cref{sec:exp}), and where we suspected residual bias we repeated it with
independent seeds; at $\beta=40$ this rejected an earlier $12{,}000$-step model whose deviation
reproduced at $-1.2\times10^{-5}$ over five draws ($13\sigma$) and confirmed the model reported
here at $+1.5\times10^{-6}$ over four draws ($0.7\sigma$).

\end{document}